\pdfoutput=1

\documentclass[fleqn]{sigplanconf}

\usepackage{amsmath}
\usepackage{amssymb}
\usepackage{amsthm}
\usepackage{url}
\usepackage[T1]{fontenc} 
\usepackage[utf8]{inputenc}
\usepackage{fixltx2e} 
\usepackage{ifthen} 
\usepackage{listings}
\usepackage{microtype}
\usepackage[usenames,dvipsnames,table]{xcolor}
\usepackage{multicol} 
\usepackage{stmaryrd}
\usepackage{mathpartir}
\usepackage{verbatim} 
\usepackage[para]{footmisc} 
\usepackage[caption=false]{subfig} 
\usepackage{pdfpages} 
\usepackage{pifont} 

\usepackage[inline]{enumitem}

\setcitestyle{notesep={~}}
\let\oldcite\cite
\renewcommand\cite[2][]{%
\ifx\\#1\\%
  \oldcite{#2}%
\else%
  \oldcite[(#1)]{#2}%
\fi}

\makeatletter
\newcommand{\laterdef}[2]{%
  \protected@write\@auxout{}{\gdef\string#1{#2}}}
\makeatother

\begingroup
\lccode`\~=`\ %
\lowercase{%
  }%
\endgroup
\newcounter{word}

\newcommand\reportwordcount[2]{%
  \footnotesize\hfill [\arabic{word} words; FK reading ease #2]%
  \laterdef{#1}{\theword}
}

\newcommand\ColdWORDS{276}
\newcommand\CnewWORDS{80}
\newcommand\OoldWORDS{391}
\newcommand\OnewWORDS{89}

\usepackage{listings}
\usepackage{tikz}
\usetikzlibrary{matrix,fit,calc,positioning,backgrounds}
\tikzset{event/.style={draw=none, inner sep=0.4mm}}
\tikzset{thd/.style={draw,dotted, rounded corners=1mm, inner sep=1mm}}
\tikzset{wg/.style={draw,dotted, rounded corners=1mm, inner sep=0.7mm}}
\tikzset{dv/.style={draw,dotted, rounded corners=1mm, inner sep=1.3mm}}

\pgfdeclarelayer{thdlayer}
\pgfdeclarelayer{wglayer}
\pgfdeclarelayer{dvlayer}
\pgfsetlayers{dvlayer,wglayer,thdlayer,main}

\usepackage{pgfplots}
\usepackage{pgfplotstable}

\newcommand{\calcrowmean}{
    \def\rowmean{0}
    \pgfmathparse{\pgfkeysvalueof{/pgfplots/table/summary statistics/end index}-\pgfkeysvalueof{/pgfplots/table/summary statistics/start index}+1}
    \edef\numberofcols{\pgfmathresult}
    \pgfplotsforeachungrouped \col in {\pgfkeysvalueof{/pgfplots/table/summary statistics/start index},...,\pgfkeysvalueof{/pgfplots/table/summary statistics/end index}}{
        \pgfmathparse{\rowmean+\thisrowno{\col}/\numberofcols}
        \edef\rowmean{\pgfmathresult}
    }
}
\newcommand{\calcstddev}{
    \def\rowstddev{0}
    \calcrowmean
    \pgfplotsforeachungrouped \col in {\pgfkeysvalueof{/pgfplots/table/summary statistics/start index},...,\pgfkeysvalueof{/pgfplots/table/summary statistics/end index}}{
        \pgfmathparse{\rowstddev+(\thisrowno{\col}-\rowmean)^2/(\numberofcols-1)}
        \edef\rowstddev{\pgfmathresult}
    }
    \pgfmathparse{sqrt(\rowstddev)}
}
\newcommand{\calcstderror}{
    \calcrowmean
    \calcstddev
    \pgfmathparse{sqrt(\rowstddev)/sqrt(\numberofcols)}
}

\pgfplotstableset{
    summary statistics/start index/.initial=1,
    summary statistics/end index/.initial=10,
    create col/mean/.style={
        /pgfplots/table/create col/assign/.code={
            \calcrowmean
            \pgfkeyslet{/pgfplots/table/create col/next content}\rowmean
        }
    },
    create col/standard deviation/.style={
        /pgfplots/table/create col/assign/.code={
            \calcstddev
            \pgfkeyslet{/pgfplots/table/create col/next content}\pgfmathresult
        }
    },
    create col/standard error/.style={
        create col/assign/.code={
            \calcstderror
            \pgfkeyslet{/pgfplots/table/create col/next content}\pgfmathresult
        }
    }
}

\newif\iftodos
\todostrue 

\newif\ifbarriers
\definecolor{hlcolor}{HTML}{FFFF99}
\newcommand\mhl[1]{\ifbarriers
#1\else\fboxsep=0pt\colorbox{hlcolor}{\strut $#1$}\fi}
\newcommand\hl[1]{\fboxsep=0pt\colorbox{hlcolor}{\strut #1}}

\let\oldparagraph=\paragraph
\renewcommand\paragraph[1]{\oldparagraph{#1.}}

\makeatletter
\setvspace{\@paragraphaboveskip}{4pt plus 2pt minus 2pt}
\usepackage{suffix}
\WithSuffix\newcommand\paragraph*[1]{\@startsection{paragraph}{4}{0pt}{\@paragraphaboveskip}{0pt}{\normalsize
\bfseries \if \@times \itshape \fi}{#1\,} \hspace*{0pt}}
\makeatother

\newcommand\tstack[2][l]{\begin{tabular}[t]{@{}#1@{}}#2\end{tabular}}

\newcommand\stack[2][l]{{\renewcommand\arraystretch{1}\begin{array}[t]{@{}#1@{}}#2\end{array}}}

\newcommand\var[1]{\mathit{#1}}

\newcommand\evtlbl[1]{\mbox{#1:~}}
\newcommand\irreflexive{\mathrm{irr}}
\newcommand\acyclic{\mathrm{acy}}
\newcommand\isempty{\mathrm{empty}}
\newcommand\consistent{\mathrm{consistent}}
\newcommand\faulty{\mathrm{faulty}}
\newcommand\allowed{\mathrm{allowed}}

\newcommand\memord{\mathit{memory\mhyphen{}order}}
\newcommand\goodsc{\mathit{good\mhyphen{}sc}}

\newcommand\sccondI{sc\mhyphen{}all}
\newcommand\sccondII{sc\mhyphen{}dv}

\newcommand\IF{\mathbf{if}}
\newcommand\THEN{\mathbf{then}}
\newcommand\ELSE{\mathbf{else}}

\newcommand\INSld{\texttt{LD}}
\newcommand\INSst{\texttt{ST}}
\newcommand\INSincl[1]{\texttt{INC$_{\texttt{L#1}}$}}
\newcommand\INSflushl[1]{\texttt{FLU$_{\texttt{L#1}}$}}
\newcommand\INSinvall[1]{\texttt{INV$_{\texttt{L#1}}$}}

\newcommand\citeline[5]{%
  \ifthenelse{\equal{#2}{#4}}{%
    \ifthenelse{\equal{#3}{#5}}{%
      \cite[#2/#3]{#1}%
    }{%
      \cite[#2/#3--#5]{#1}%
    }%
  }{%
    \cite[#2/#3--#4/#5]{#1}%
  }%
}
\newcommand\citecl{\citeline{opencl21}}

\newcommand\citec[3][]{%
  \ifthenelse{\equal{#3}{}}{%
    \cite[\S#2]{c11}%
  }{%
    \cite[\S#2:#3]{c11}%
  }%
}

\definecolor{colormo}{HTML}{3366CC}
\definecolor{colorrf}{HTML}{FF0000}
\definecolor{colorrb}{HTML}{006600}
\definecolor{colorhb}{HTML}{660000}
\definecolor{colorS}{HTML}{990099}

\tikzset{
  edgemo/.style={colormo,-latex},
  edgerf/.style={colorrf,-latex},
  edgerb/.style={colorrb,dashed,-latex},
  edgeS/.style={colorS, -latex},
  edgesb/.style={black, -latex},
  edgethd/.style={black, latex-latex}
}

\newcommand\axiom[1]{\mbox{\upshape\sffamily #1}}
\newcommand\oaxiom[1]{\mbox{\upshape\sffamily O-#1}}

\newcommand\pgap{\hspace{-0.1em}}
\renewcommand\parallel{\mathbin{{\mid}\pgap{\mid}}}
\renewcommand\bigparallel{\mathop{{\mid}\pgap{\mid}}}
\newcommand\tripleparallel{\mathbin{{\mid}\pgap{\mid}\pgap{\mid}}}
\newcommand\bigtripleparallel{\mathop{{\mid}\pgap{\mid}\pgap{\mid}}}
\newcommand\quadparallel{%
  \mathbin{{\mid}\pgap{\mid}\pgap{\mid}\pgap{\mid}}}
\newcommand\bigquadparallel{%
  \mathop{{\mid}\pgap{\mid}\pgap{\mid}\pgap{\mid}}}

\newcommand\herd{{\scshape Herd}}
\newcommand\Herd{{\scshape Herd}}
\newcommand\Cppmem{Cppmem}
\newcommand\CDSChecker{CDSChecker}
\newcommand\cat{{\ttfamily .cat}}

\newcommand\axSpartial{\axiom{S\textsubscript{partial}}}
\newcommand\axSsimp{\axiom{S\textsubscript{simp}}}
\newcommand\axOSsimp{\oaxiom{S\textsubscript{simp}}}
\newcommand\axOSscoped{\oaxiom{S\textsubscript{scoped}}}

\newcommand\morlx{{\tt RLX}}
\newcommand\moacq{{\tt ACQ}}

\newcommand\morel{{\tt REL}}
\newcommand\moar{{\tt AR}}
\newcommand\mosc{{\tt SC}}

\newcommand\swi{{\tt WI}}
\newcommand\swg{{\tt WG}}
\newcommand\sdv{{\tt DV}}
\newcommand\sall{{\tt ALL}}

\newcommand\atomic{{\rm atomic}}
\newcommand\nonatomic{{\rm non\mhyphen{}atomic}}
\newcommand\Label{\mathcal L}
\newcommand\evW{{\rm W}}
\newcommand\evWna{{\rm W}_{\rm na}}
\newcommand\evR{{\rm R}}
\newcommand\evRna{{\rm R}_{\rm na}}
\newcommand\evRMW{{\rm RMW}}
\newcommand\evF{{\rm F}}
\newcommand\evFL{{\rm F}_{\rm L}}
\newcommand\evFG{{\rm F}_{\rm G}}
\newcommand\evFGL{{\rm F}_{\rm GL}}

\everymath{\it}\everydisplay{\it}
\def\mathcodes#1{\mathcode`#1=\numexpr\mathcode`#1-"7000\relax 
   \ifx#10\else\expandafter\mathcodes\fi}
\mathcode`\;="8000 
{\catcode`\;=\active \gdef;{\mathbin{\semicolon}}}
\mathchardef\semicolon="603B

\mathchardef\mhyphen="2D

\catcode`"=\active
\def"#1"{\mbox{\tt #1}}

\renewcommand\jot{0.5pt}

\newcommand\eqdef{\overset{{\scriptstyle\rm def}}{=}}

\title{Overhauling SC Atomics in C11 and OpenCL}

\authorinfo{Mark Batty}{University of Kent, UK}{m.j.batty@kent.ac.uk}
\authorinfo{Alastair F. Donaldson}{Imperial College London, UK}{alastair.donaldson@imperial.ac.uk}
\authorinfo{John Wickerson}{Imperial College London,
UK}{j.wickerson@imperial.ac.uk\vspace*{10mm}}

\usepackage[firstpage]{draftwatermark}
\SetWatermarkText{\hspace*{8in}\raisebox{3.5in}{\includegraphics[scale=0.1]{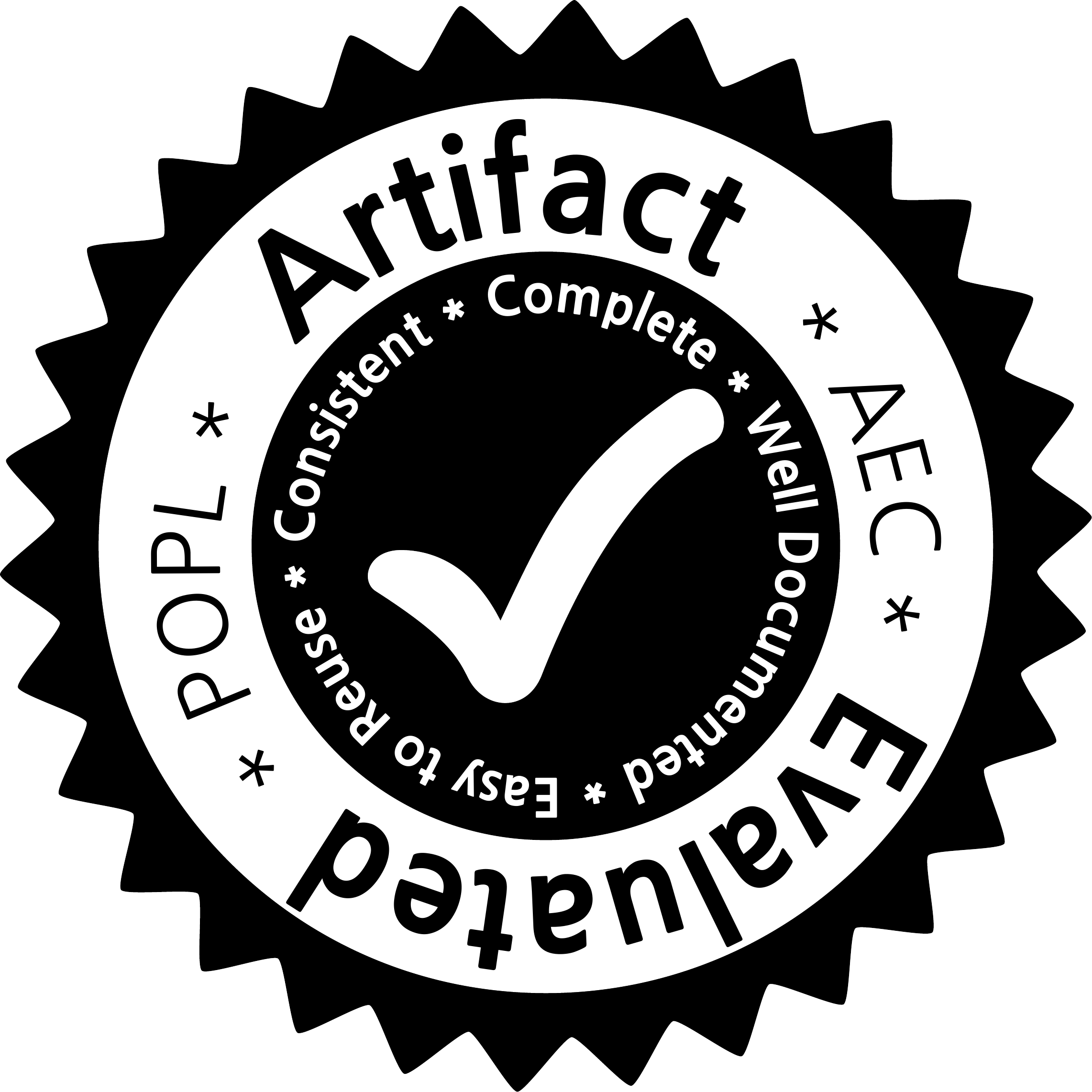}}}
\SetWatermarkAngle{0}

\allowdisplaybreaks[1]

\usepackage{flushend}

\usepackage{framed}
\definecolor{shadecolor}{HTML}{EEEEEE}

\newtheorem{theorem}{Theorem}
\newtheorem{lemma}{Lemma}
\theoremstyle{definition}
\newtheorem{definition}[lemma]{Definition}
\newtheorem{remark}[lemma]{Remark}
\newtheorem{example}{Example}

\newenvironment{commentary}{\begin{proof}[Commentary]\phantom\qedhere}{\end{proof}}
\newenvironment{Example}{\begin{snugshade}\begin{example}}{\end{example}\end{snugshade}}
\begin{document}

\toappear{}

\maketitle

\begin{abstract} 
Despite the conceptual simplicity of sequential consistency (SC), the
semantics of SC atomic operations and fences in the C11 and OpenCL
memory models is subtle, leading to convoluted prose descriptions that
translate to complex axiomatic formalisations. We conduct an overhaul
of SC atomics in C11, reducing the associated axioms in both number
and complexity. A consequence of our simplification is that the SC
operations in an execution no longer need to be totally ordered. This
relaxation enables, for the first time, efficient and exhaustive
simulation of litmus tests that use SC atomics. We extend our improved
C11 model to obtain the first rigorous memory model formalisation for
OpenCL (which extends C11 with support for heterogeneous many-core
programming). In the OpenCL setting, we refine the SC axioms still
further to give a sensible semantics to SC operations that employ a
`memory scope' to restrict their visibility to specific threads. Our
overhaul requires slight strengthenings of both the C11 and the OpenCL
memory models, causing some behaviours to become disallowed. We argue
that these strengthenings are natural, and that all of the formalised
C11 and OpenCL compilation schemes of which we are aware (Power and
x86 CPUs for C11, AMD GPUs for OpenCL) remain valid in our revised
models. Using the \herd{} memory model simulator, we show that our
overhaul leads to an exponential improvement in simulation time for
C11 litmus tests compared with the original model, making
\emph{exhaustive} simulation competitive, time-wise, with the
\emph{non-exhaustive} \CDSChecker{} tool.
\end{abstract}



%

\section{Introduction}

\paragraph{Atomics and memory models} C11 and OpenCL both define a
collection of \emph{atomic operations}, or `atomics', which can be used
by experts to program high-performance, lock-free algorithms in a
portable manner. Atomics accept a \emph{memory order} parameter, which
controls the exposure of certain \emph{relaxed memory} behaviours that
modern CPUs and GPUs natively exhibit.

The C11 and OpenCL specifications~\cite{c11, opencl21} define the
semantics of atomics via axiomatic \emph{memory models}; that is, sets
of rules that govern the reading and writing of shared memory
locations. These memory models are complex, stretching to about 19~and
30~pages, respectively, of convoluted prose. This complexity makes it
extremely challenging to reason about the correctness of programs that
are written in, and compilers that implement, these languages. 

Correctness in any relaxed memory setting is notoriously evasive;
indeed, the subtleties of relaxed memory have previously led to
confirmed bugs in language specifications~\cite{batty+11, batty+15},
deployed processors~\cite{alglave+10}, compilers~\cite{sevcik+08,
morriset+13} and vendor-endorsed programming guides~\cite{alglave+15}.
The importance of correctness in the context of C11 is
well-known. Correctness is just as crucial in OpenCL, which is an open
standard for \emph{heterogeneous programming} that is developed and
supported by major hardware vendors such as Altera, AMD, ARM, Intel,
Nvidia, Qualcomm and Xilinx. OpenCL is a key player in the recent
drive to exploit GPUs and FPGAs in general-purpose computing,
including in safety-critical domains such as medical
imaging~\cite{steuwer+13} and autonomous
navigation~\cite{khronos-group-news-archives14}.

We seek in our work to tame the complexity of these memory models
through \emph{formalisation}.

\paragraph*{The C11 memory model} has been formalised by several
researchers, in varying degrees of completeness, and with varying
degrees of fidelity to the standard~\cite{batty+11, vafeiadis+13,
alglave+14}. These formalisation efforts have proved fruitful; they
have, for instance, enabled the construction of simulators that
automatically explore the allowed behaviours of small C11 programs
(called \emph{litmus tests})~\cite{batty+11, blanchette+11, norris+13,
alglave+14}, underpinned the design of program logics for specifying
and verifying C11 programs~\cite{vafeiadis+13, turon+14}, and they provide
a firm foundation for ongoing debate about the design of the C11
memory model itself~\cite{vafeiadis+15, batty+15}.

\paragraph*{The OpenCL memory model} (introduced in version 2.0 of the
standard) has received comparatively little academic attention, with
the notable exception of the work of Gaster et al.~\cite{gaster+15},
which we discuss further in \S\ref{sec:related}. OpenCL provides a
framework for CPU programs to delegate the execution of
massively-parallel \emph{kernel} functions, written in a variant of C,
to one or more \emph{accelerator devices}, such as GPUs or FPGAs. Threads that
execute these kernels are organised into a hierarchy: threads\footnote{Threads in OpenCL are also called \emph{work-items}.} are
grouped into \emph{work-groups}, and work-groups are grouped by
device. The OpenCL memory model is broadly similar to that of C11, but
is extended with features such as \emph{memory regions} (which contain
locations that are accessible only to a certain subtree of the thread
hierarchy), and \emph{memory scopes} (which, when applied to an atomic
operation, confine its visibility to a certain subtree of threads).

\paragraph{SC atomics}
Our work is distinguished by its focus on the \emph{sequentially
consistent} (SC) fragment of these memory models; that is, the
semantics of atomics whose memory order is "memory\_order\_seq\_cst".
The chief guarantee provided by this memory order is that all SC
atomics in a given execution will execute in some order (say, $S$) on
which all threads mutually agree. Note that these memory models do not
\emph{construct} $S$; they merely postulate the existence of a suitable
$S$.

Sequential consistency is known for its simplicity~\cite{lamport79},
and indeed, any C11 or OpenCL program using \emph{exclusively} SC
atomics would enjoy a simple interleaving semantics. However, when
combined with the more relaxed memory orders that C11 and OpenCL also
provide, the semantics of SC atomics becomes highly complex, and it is
this complexity that we tackle in this paper.

SC atomics are in widespread use, partly because the SC memory order
is used when no other is specified, and partly because programmers are
routinely advised to use SC atomics prior to optimising their code
with the more relaxed memory orders~\cite[p.~221]{williams12}.
Algorithms that make use of SC atomics include Dekker's mutual
exclusion algorithm~\cite{dijkstra02}, and more generally,
multiple-producer-multiple-consumer algorithms that require every
consumer to observe the actions of every producer in the same
order.\footnote{\url{http://en.cppreference.com/w/cpp/atomic/memory_order}} As such, it is important that the semantics of SC
atomics is clear to programmers, to allow smooth transitioning
between the exclusive use of SC (for ease of reasoning) to a mixture
of SC and weaker-than-SC atomics (for performance optimisation).

In theory, SC atomics can be avoided by replacing them
with mutex-protected non-atomic operations (and simple spinlock
mutexes can be implemented using just release and acquire
atomics~\cite[p.~111]{williams12}).  In practice, support for SC atomics is
non-negotiable if software is to make use of concurrency
libraries. This is because the aforementioned replacement of SC
atomics must be performed throughout the entire program -- in both
library code and client code alike -- and with the same mutex variable
for every operation. Moreover, programs that make extensive use of
spinlocks could prove less efficient than those that rely on native SC
atomics, and accidental misuse of locks may lead to deadlock.

\subsection{Main Contributions}

Our work aims to provide clearer, simpler foundations for reasoning
about C11, enabling a clean extension to OpenCL for heterogeneous
programming, and facilitating efficient simulation.

\paragraph{1. Overhauling SC atomics in C11 (\S\ref{sec:sc})}
The C11 specification devotes around \ColdWORDS{} words to explaining
the semantics of SC atomics. In our work, we have translated these
words into mathematical axioms, carefully strengthened these axioms
(without imposing unreasonable demands on the compiler), and then
refactored them so that they are expressed as simply as possible. Our
revised text
\begin{itemize}[label=\checkmark]

\item is shorter (requiring just \CnewWORDS{} words in the same prose style),

\item is simpler (because it reduces seven axioms to just one), and

\item is amenable to more efficient simulation (see below).

\end{itemize}

Supporting the revised text is a provably-equivalent model that avoids
the need to postulate the total order $S$. Instead, the model
constructs a partial order on SC operations, preserving only the edges
of $S$ that can affect program behaviours. The enumeration of all
candidate $S$ relations is one of the most expensive tasks for memory
model simulators like \herd{}; by reducing $S$ to a partial order, we
can dramatically improve simulation performance.

\paragraph{2. Overhauling SC atomics in OpenCL (\S\ref{sec:opencl_sc})}
Our simplifications to the rules governing SC atomics in C11 can be
carried over directly to OpenCL, where the same three benefits listed
above can be reaped. In the OpenCL setting, however, there is an
additional complexity in the semantics of SC atomics. Specifically,
the total order $S$ in which all of a program's SC atomics execute is
only guaranteed to exist when one of two conditions holds: either all
SC atomics in the program's execution use the widest-possible memory
scope and only access memory shared between devices, or all SC atomics
have their memory scope limited to the current device and never access
memory shared between devices. We find that this semantics is
unhelpful to programmers, because if \emph{any} SC atomic violates
these conditions, then \emph{no} SC atomic is guaranteed to have
semantics stronger than acquire/release; this may lead to additional
behaviours not anticipated by the programmer.
The semantics is
simultaneously unhelpful to compiler-writers: a loop-hole that we
discovered in the second condition above means that even device-scoped
SC atomics must be implemented using expensive inter-device
synchronisation.

We have amended the rules that govern SC atomics in OpenCL, so that
the SC guarantees do not vanish immediately in the presence of a
differently-scoped SC atomic somewhere in the program, but instead
degrade gracefully. Our revised rules
\begin{itemize}[label=\checkmark]

\item are shorter and simpler (we can replace \OoldWORDS{}
words in the specification with \OnewWORDS{} words in the same prose style),

\item enable new programming patterns in OpenCL (such as programs that
use SC atomics in a natural manner, yet a manner that violates the
overly restrictive conditions above), 

\item let device-scoped SC atomics be efficiently implemented, and

\item improve the compositionality of OpenCL semantics, and hence the
ability to write concurrency libraries (because the
behaviour of SC atomics no longer depends on unstable, global conditions).

\end{itemize}

\paragraph{3. Proving the implementability of our revised models
(\S\ref{sec:c11_soundness}, \S\ref{sec:opencl_implementability})} 
Our improvements to the SC axioms in the C11 and OpenCL memory models
hinge on slight strengthenings of the models; that is, tweaking some
of the axioms so that fewer executions are allowed. This increases the
demands on compilers that implement these memory models, so it is
important to check that our changes do not invalidate existing
compilation schemes. To this end, we prove that all of the formalised
C11 compilation schemes of which we are aware (namely, those for
Power~\cite{batty+12} and x86~\cite{batty+11} machines) remain sound
after our changes, and we argue informally that our OpenCL changes
preserve the soundness of the only formalised OpenCL compilation
scheme (namely, that for AMD GPUs~\cite{wickerson+15a}).

\subsection{Supporting Contributions}

In order to justify the claims we make in our main contributions, we
have established several supporting artefacts, which we believe are
also valuable in their own right.

\paragraph{4. Formalising the OpenCL memory model
(\S\ref{sec:opencl_sc})} The OpenCL specification contains numerous
ambiguities, omissions and inconsistencies, which makes it a shaky
structure upon which to build an argument about the correctness of an
OpenCL program or compiler. The lack of clarity may lead programmers
and compiler-writers to cautiously opt for low-efficiency
implementations that are easier to guarantee correct. Moreover, there
are instances where the OpenCL specification authors have made
unnecessarily conservative, programmer-unfriendly decisions in the
design of rules for memory consistency. We provide the first
mechanised formalisation of the OpenCL memory model. Our formalisation
serves to clarify the specification, and can henceforth be used to
underpin future program logics for verifying OpenCL kernels, and to
inform further refinements to the memory model.\footnote{Indeed, we
have already built upon our formalisation in another piece of work
that investigates a proposed extension to the OpenCL memory
model~\cite{wickerson+15a}.} In particular, we use our rigorous memory
model to show that the design decisions of the specification can be
made less conservative, offering programmers more flexibility, without
placing any additional burden on efficient implementation of the
language.

\paragraph{5. Formalising the memory models in \cat{}
(\S\ref{sec:c11.cat}, \S\ref{sec:opencl.cat})} 

We have encoded the C11 and OpenCL memory models in the \cat{}
framework~\cite{alglave+14}. Previous formalisations of the C11 memory
model exist, in Isabelle~\cite{batty+11}, Lem~\cite{batty+12} and
Coq~\cite{vafeiadis+13}; here we contribute the first version in
\cat{}. 
We conduct our development work in the \cat{} language because it is
the native input format to the \herd{} memory model simulator, which
has a proven record of efficiently simulating a range of CPU
machine-level memory models~\cite[\S8.3]{alglave+14}.

\paragraph{6. Extending the \herd{} memory model simulator (\S\ref{sec:herd})}
During memory modelling work, tool support for simulating alternative
memory models against litmus tests is invaluable. \Herd{} is able to
simulate any memory model expressed as a \cat{} file, but in its
original incarnation, it supported only machine-level models of
CPUs~\cite{alglave+14} and GPUs~\cite{alglave+15}. To explore our
proposed changes to the C11 and OpenCL memory models, we extended
\herd{} with a module for generating executions of C11 and OpenCL
programs, and support for language-level memory models that
incorporate `undefined behaviour' (a notion that is absent from
machine-level models). This involved adding around 8000 lines to the
original \herd{} codebase.\footnote{As estimated by \texttt{git log}.}
All of the examples in this paper have been automatically checked with
\Herd{}. Using \herd{}, we have evaluated the impact of our changes to
the SC axioms, and found an exponential improvement in simulation
performance.

\paragraph{Online material} Our companion webpage
provides instructions for downloading \herd{} and our \cat{} formalisations~\cite{overhauling-companion}.

\section{The C11 Memory Model in \cat{}}
\label{sec:c11.cat}

This section describes formally the current C11 memory model.

The semantics of multi-threaded C11 programs is formalised in two
stages; the first concerning the thread-local semantics, and the
second capturing the memory model. Roughly speaking, the first stage
takes as input a C11 program and calculates its set of
\emph{executions} (\S\ref{sec:c11_programs},
\S\ref{sec:c11_executions}); the second stage then compares each
execution to the memory model to determine which executions are
actually allowed (\S\ref{sec:c11_axioms}).

There exist several prior formalisations of the C11 memory
model~\cite{batty+11, vafeiadis+13, alglave+14}. The novelty of this
section is the first comprehensive formalisation of the model in the
\cat{} framework~\cite{alglave+14}, which enables the use of the
efficient \herd{} simulator~\cite{alglave+14}. For reasons of space,
and because they are orthogonal to the thrust of our contributions, we
omit our treatment of the `consume' memory order, unsequenced races
and C11 locks from the paper. However, our \cat{}-based formalisation
fully accounts for these features, and is provided on our companion webpage~\cite{overhauling-companion}.

\subsection{C11 Programs}
\label{sec:c11_programs}

A C11 program manipulates a set of shared memory locations.

\begin{definition}[Memory locations]
\label{def:c11_locs}
Each memory location is declared with either a \emph{non-atomic} or an
\emph{atomic} type. That is, $type(l) \in \{\atomic,\nonatomic\}$
for every memory location $l$. 
%
\end{definition}

\begin{definition}[Structure of C11 programs]
\label{def:c11_syntax}
We consider C11 programs of the form $P = \bigparallel_{t\in T} p_t$,
where $T$ is a set of thread identifiers, $p_t$ is a piece of
sequential code, and $\parallel$ is parallel composition. (This static
form of parallelism is a simplification of the dynamic thread creation
that C11 actually provides.)
\end{definition}
Atomic locations can be accessed via \emph{atomic operations}; these
include reads, writes, and read-modify-writes (RMWs). C11 also defines
\emph{fence} operations. Atomic operations and fences expose the
programmer to relaxed memory behaviours; which behaviours are exposed
is controlled by the operation's \emph{memory order} parameter.

\begin{definition}[Memory orders] 
The available memory orders in C11 are:
\[
\begin{array}{r@{~}ll}
o ::= & \morlx & \text{(relaxed)} \\
 \mid & \moacq & \text{(acquire, only for reads/RMWs)} \\
 \mid & \morel & \text{(release, only for writes/RMWs)} \\
 \mid & \moar  & \text{(acquire+release, only for RMWs)} \\ 
 \mid & \mosc  & \text{(sequentially consistent, the default).} 
\end{array}
\] 
\end{definition}

\begin{Example}[A C11 program]
\label{ex:c11_prog}
We give below a contrived C11 program that operates on two atomic
locations, "x" and "y", using atomic "store" and "load" operations
with a variety of memory orders.
\begin{center}
\scriptsize
\begin{tabular}{@{}l@{~}||@{~}l@{~}||@{~}l@{~}||@{~}l@{}}
\multicolumn{4}{c}{\texttt{atomic\_int *x; atomic\_int *y;}} \\
\texttt{store(x,1,$\morlx$);} &
\texttt{r1=load(x,$\morlx$);} & 
\texttt{store(x,2,$\mosc$);} &
\texttt{store(y,1,$\mosc$);} 
\\
& 
\texttt{r2=load(x,$\morlx$);} & 
\texttt{r3=load(y,$\mosc$);} & 
\texttt{r4=load(x,$\mosc$);}
\end{tabular}
\end{center}
\end{Example}


\subsection{C11 Executions}
\label{sec:c11_executions}

The C11 memory model is defined in terms of program \emph{executions}.
An execution $X$ takes the form of a mathematical graph, where each
node $e\in E$ is labelled with a run-time memory event (see
Def.~\ref{def:c11_event_labels}), and the edges connect events
performed by the same thread in program order. In other words, an
execution is a partial order over a set $E$ of events, and can be
thought of as a `concurrent trace'.

\begin{definition}[Event labels]
\label{def:c11_event_labels}
Each event's label characterises the kind of instruction that
gave rise to the event, and incorporates up to four attributes, as
listed in the first five columns of the following table:
\[
\begin{array}{l@{~}l@{}cccc@{}l|cccc}
kind & & loc & rval & wval & ord & & R & W & F & A \\ \hline
\evWna & (&l,&   & v,& &)& & \checkmark & & \\
\evW   & (&l,&   & v,& o&)& & \checkmark & & \checkmark \\
\evRna & (&l,& v,&   & &)& \checkmark & & & \\ 
\evR   & (&l,& v,&   & o&)& \checkmark & & & \checkmark\\ 
\evRMW & (&l,& v,& v',&o&)& \checkmark & \checkmark & & \checkmark\\ 
\evF   & (&  &   &   & o&)& & & \checkmark & \checkmark
\end{array}
\]
The labels represent (reading down): non-atomic writes, atomic writes,
non-atomic reads, atomic reads, RMWs (which are always atomic), and
memory fences. Where relevant, labels contain (reading across): the
location being accessed, the value being read, the value being
written, and the memory order specified by the programmer. A
$\checkmark$-mark on the right-hand side of the table indicates that
an event with this label belongs to the set $R$ (resp.~$W$, $F$, $A$)
of events that read (resp.~write, are a fence, are atomic). Let
$\Label$ denote the set of labels.
\end{definition}
\begin{definition}[Executions] 
\label{def:c11_executions}
An execution is a tuple
$X = (E, I, lbl,$ $thd,$ $sb)$ with the following components. 
\begin{itemize}
\item $E$ is a set of event identifiers.

\item $lbl \in E \rightarrow \Label$ associates each event with a
label. For each event $e$, $loc(e)$ projects the $loc$ attribute
of $lbl(e)$ (if applicable); $rval(e)$, $wval(e)$ and $ord(e)$
provide similar projections.

\item $I\subseteq E$ is a set of \emph{initial} events. Every initial
event $e \in I$ is a non-atomic write of zero; that is,
$kind(e) = \evWna$ and $wval(e) = 0$. Moreover,
there is exactly one initial event per location
.

\item $thd \subseteq (E\setminus I)^2$ is an equivalence relation on
non-initial events that relates events from the same thread.

\item $sb\subseteq thd$ is the \emph{sequenced-before} relation: a
strict partial order (i.e., irreflexive and transitive) between events from
the same thread, that captures the program order.
\end{itemize}
\end{definition}

Let $\mathbb X$ be the set of all executions. Next, we define a number
of derived sets and relations over the events of an execution that
will prove useful in describing the memory model.

\begin{definition}[Derived sets and relations]
In the context of an execution $(E,I,lbl,thd,sb)$, we define the
relation $=_{loc}$ as
$\{(e,e') \in (E \setminus F)^2 \mid loc(e) = loc(e')\}$; it holds
between non-fence events that access the same location. The relation
$=_{val}$, defined as
$\{(e,e') \in W \times R \mid wval(e) = rval(e')\}$, holds when the
first event writes the value that the second reads. For each memory
order $o\in\{\morlx,\moacq,\morel,\moar,\mosc\}$, we abbreviate the
set $\{e\in A\mid ord(e) = o\}$ as just $o$. We also define the set
$nal = \{e\in E\setminus F \mid type(loc(e)) = \nonatomic\}$ of events
that access a non-atomic location.
\end{definition}

\begin{Example}[A C11 execution] 
\label{ex:c11_exec}
The diagram below depicts one execution of the program given in
Example~\ref{ex:c11_prog}. The initial events, $a$ and $b$, are placed
above the events of the four parallel threads. Reflexive and
transitive edges are elided, and derived relations are not shown.  
%
\begin{center}
\begin{tikzpicture}[inner sep=1pt]
\node[event, anchor=west](h) at (1.5,1.7)
{\evtlbl{$a$}$\evWna({\tt x},0)$};

\node[event, anchor=west](i) at (4.5,1.7) 
{\evtlbl{$b$}$\evWna({\tt y},0)$};

\node[event, anchor=west](a) at (-0.3,1) 
{\evtlbl{$c$}$\evW({\tt x},1,\morlx)$};

\node[event, anchor=west](b) at (2,1) 
{\evtlbl{$d$}$\evR({\tt x},1,\morlx)$};

\node[event, anchor=west](c) at (2,0.2) 
{\evtlbl{$e$}$\evR({\tt x},2,\morlx)$};

\node[event, anchor=west](d) at (4.1,1) 
{\evtlbl{$f$}$\evW({\tt x},2,\mosc)$};

\node[event, anchor=west](e) at (4.1,0.2) 
{\evtlbl{$g$}$\evR({\tt y},0,\mosc)$};

\node[event, anchor=west](f) at (6.2,1) 
{\evtlbl{$h$}$\evW({\tt y},1,\mosc)$};

\node[event, anchor=west](g) at (6.2,0.2) 
{\evtlbl{$i$}$\evR({\tt x},1,\mosc)$};

\foreach \i/\j in {d/e, f/g, b/c}
\draw[edgesb] ([xshift=1mm]\i.south) to[auto,pos=0.4]
node{$sb$} ([xshift=1mm]\j.north -| \i.south);

\foreach \i/\j in {d/e, f/g, b/c}
\draw[edgethd] ([xshift=-1mm]\i.south) to[auto,swap,pos=0.4]
node{$thd$} ([xshift=-1mm]\j.north -| \i.south);
\end{tikzpicture}
\end{center}
\end{Example}

\paragraph{Basic executions} The first stage of the C11 semantics
translates a program into a set of
executions called its \emph{basic} set.\footnote{This set is sometimes called the
`pre-executions'~\cite{batty+11} or the `opsems'~\cite{vafeiadis+15}.}
Each execution in this set is compatible with the instructions of the
individual threads, but the set is constructed without considering the
behaviour of shared memory, so it provides an over-approximation of
the executions that will ultimately be allowed to happen once the
whole program and the memory model are taken into account. For
instance, the execution in Example~\ref{ex:c11_exec} is a basic
execution of the program in Example~\ref{ex:c11_prog}: the values of
the write events correspond to the program text, but the values of the
read events are arbitrary and the basic set of all executions ranges
over all choices. We do not define formally how the basic executions
are constructed, and simply assume their existence for any program we
wish to consider. Practical tools such as \herd{} and
\Cppmem{}~\cite{batty+11} implement this construction as part of
litmus test simulation; the construction is investigated formally in
ongoing work by Memarian et al.

\paragraph{Candidate executions} The second stage of the C11
semantics, which is the focus of this paper, takes as input a
program's basic execution set and returns the set of \emph{allowed
executions}. In order to build the allowed executions, we employ an
intermediate structure called a \emph{candidate execution}, which
extends an execution with a \emph{witness} that comprises three
additional relations, called $rf$ (reads-from), $mo$ (modification
order) and $S$ (sequential consistency order).

\begin{definition}[Candidate executions] 
\label{def:candidate_exec}
A candidate execution is a pair $(X,w)$ where $X=(E,I,lbl,thd,sb)$ is
an execution, and $w=(rf,mo,S)$ is a witness comprising three
relations $rf,mo,S\subseteq E^2$. A candidate execution is
well-formed, written $wf(X,w)$, if: 
\begin{itemize}

\item the reads-from relation links write events to read events,
such that every read observes exactly one write, and the
locations and values match; that is,
\begin{gather}
\tag{\axiom{WfRf}}
\left.\begin{array}{@{}p{5.7cm}@{}}$\forall e\in R \ldotp \exists! e'\in W\ldotp (e',e) \in rf$ \\[1mm]
and~~~ $rf \subseteq ({=_{loc}} \cap {=_{val}})$\end{array}\right\}
\end{gather}
where $\exists!$ means `exists unique';

\item the modification order
relates, in a strict total order, all and only those events that write
to the same atomic location; that is,
\begin{gather}
\tag{\axiom{WfMo}}
\left.\begin{array}{@{}p{5.7cm}@{}} $(mo \cup mo^{-1}) = ({=_{loc}}
\cap W^2 \setminus nal^2 \setminus id)$ \\[1mm] and~~~ $\acyclic(mo)$\end{array}\right\}
\end{gather}
where $\acyclic(r)$ means that $r$ is acyclic; and

\item the $S$ relation relates, in a
strict total order, all and only the SC events in an execution; that is,
\begin{gather}
\tag{\axiom{WfS}}
\acyclic(S)\text{~~~and~~~}(S \cup S^{-1}) = (\mosc^2 \setminus id)
\end{gather}
\end{itemize}
\end{definition}

\begin{Example}[A C11 candidate execution] 
\label{ex:c11_candidate_exec}
The diagram below extends the execution in Example~\ref{ex:c11_exec}
with a witness. We elide the $thd$ edges (each column corresponds to
one thread). 
The candidate execution is well-formed, and consistent with the axioms
of the memory model (presented next).
\begin{center}
\begin{tikzpicture}[inner sep=1pt]
\node[event, anchor=west](h) at (0.7,1.8)
{\evtlbl{$a$}$\evWna({\tt x},0)$};

\node[event, anchor=west](i) at (3.4,1.8) 
{\evtlbl{$b$}$\evWna({\tt y},0)$};

\node[event, anchor=west](a) at (-0.3,1) 
{\evtlbl{$c$}$\evW({\tt x},1,\morlx)$};

\node[event, anchor=west](b) at (2,1) 
{\evtlbl{$d$}$\evR({\tt x},1,\morlx)$};

\node[event, anchor=west](c) at (2,0) 
{\evtlbl{$e$}$\evR({\tt x},2,\morlx)$};

\node[event, anchor=west](d) at (4.1,1) 
{\evtlbl{$f$}$\evW({\tt x},2,\mosc)$};

\node[event, anchor=west](e) at (4.1,0) 
{\evtlbl{$g$}$\evR({\tt y},0,\mosc)$};

\node[event, anchor=west](f) at (6.2,1) 
{\evtlbl{$h$}$\evW({\tt y},1,\mosc)$};

\node[event, anchor=west](g) at (6.2,0) 
{\evtlbl{$i$}$\evR({\tt x},1,\mosc)$};

\foreach \i/\j in {d/e, f/g, b/c}
\draw[edgesb] ([xshift=1mm]\i.south) to[auto,pos=0.4]
node{$sb$} ([xshift=1mm]\j.north -| \i.south);

\foreach \i/\j in {d/e, f/g}
\draw[edgeS] (\i.south) to[auto,swap,pos=0.4]
node{$S$} (\j.north -| \i.south);
\draw[edgeS, inner sep=0] (e) to[auto, swap] node{$S$} (f);

\draw[edgemo] (a) to[auto, bend left=15] node{$mo$} (d);

\draw[edgemo] (h) to[auto, swap, bend right=15] node{$mo$} (a);

\draw[edgemo] (i) to[auto, swap, bend left=15] node{$mo$} (f);

\draw[edgerf] (i) to[auto, bend right=12, pos=0.25] node{$rf$} (e);

\draw[edgerf] (a) to[auto, swap, bend right=20] node{$rf$} (b);

\draw[edgerf] (d) to[auto] node{$rf$} (c);


\draw[edgerf] (a) to[auto, swap, out=300, in=188, pos=0.15] node{$rf$}
(g);
\end{tikzpicture}
\par\vspace*{-4mm}
\end{center}
\end{Example}

\subsection{C11 Axioms}
\label{sec:c11_axioms}

A candidate execution is deemed \emph{consistent} with the memory
model if it satisfies the 12 consistency axioms of
Def.~\ref{def:consistency_axioms}, which we shall build towards in
this subsection. We express the axioms using the \cat{}
language~\cite{alglave+14}, a concise language based on the
propositional fragment of Tarski's relation calculus~\cite{tarski41}.

\begin{definition}[The \cat{} language] The cat language supports the
construction of relations via: union, intersection, difference,
complement ($\neg r$), inverse ($r^{-1}$), reflexive closure ($r^?$),
transitive closure ($r^+$), and relational composition ($r_1;r_2$),
which is defined such that $(x,z)\in r_1; r_2$ if $(x,y)\in r_1$ and
$(y,z)\in r_2$ for some $y$. It also provides the syntax
$[s] = \{(e,e)\mid e \in s\}$ for the identity relation ($id$)
restricted to the set $s$. (These operators can be neatly combined to
describe paths through graphs; for instance,
$[s_1]; r_1; [s_2] ; r_2 ; [s_3]$ relates $s_1$-events to those
$s_3$-events that are reachable by following an $r_1$-edge to an
$s_2$-event and then an $r_2$-edge.) Each axiom of the memory model
must be expressed in the form of an acyclicity ($\acyclic\,r$),
irreflexivity ($\irreflexive\,r$), or emptiness ($\isempty\,r$)
constraint on some relation $r$ constructed using these operators.
\end{definition}

In order to define these axioms, we first
need to introduce several derived relations.

\begin{remark}
In the following, we justify our formal definitions by referring to
the C11 standard~\cite{c11}, using the notation \S$N$:$n$ for section
$N$, paragraph $n$. We refer to the C++11 standard~\cite{c++11},
whenever a clause was erroneously omitted from C11. (C11 inherits its
memory model from C++11). Similarly, we refer to the C++14
standard~\cite{c++14} in the case of an erroneous omission from C++11.
We include these omitted parts because doing so leads to a cleaner
model that we believe to be closer to the designers' intent.
\end{remark}

\begin{definition}[Further derived sets and relations]
\label{def:c11_further_derived}
In the context of a candidate execution $(E,I,lbl,thd,sb,rf,mo,S)$, we
define the following subsets of $E$ and relations over $E$:
\begin{eqnarray*}
acq &\eqdef& \moacq \cup \moar \cup (\mosc \cap (R\cup F)) \\
rel &\eqdef& \morel \cup \moar \cup (\mosc \cap (W\cup F)) \\
fr &\eqdef& rf^{-1} ; mo \\
Fsb &\eqdef& [F] ; sb \\
sbF &\eqdef& sb ; [F] \\
rs' &\eqdef& thd \cup (E^2 ; [R \cap W]) \\
rs &\eqdef& mo \cap rs' \setminus ((mo \setminus rs') ; mo) \\
sw &\eqdef& \stack{([rel] ; Fsb^? ; [A \cap W] ; rs^? ; rf ; {}\\{} [R\cap A] ; sbF^?; [acq]) \setminus thd}
\\
hb &\eqdef& (sb \cup (I \times \neg I) \cup
sw)^+ \\
hbl &\eqdef& hb \cap {=_{loc}} \\
vis &\eqdef& (W \times R) \cap hbl \setminus(hbl ; [W] ; hb) \\
cnf &\eqdef& ((W \times W) \cup (W \times R) \cup (R \times
W)) \cap {=_{loc}} \\
dr &\eqdef& cnf \setminus hb \setminus hb^{-1} 
\setminus A^2 \setminus thd
\end{eqnarray*}
\end{definition}

\begin{commentary} 
The set $acq$ (resp.~$rel$) contains all events that behave as an
acquire (resp.~a release).\footnote{\citec{7.17.3}{3--4},
\citec{7.17.4.1}{2}} The \emph{from-read} relation ($fr$) links each
read to all those writes that are $mo$-after the write the read
observed~\cite{alglave+14}.

The relation $rs$ captures the \emph{release sequence}, using $rs'$ as
a helper. The release sequence of $e$ comprises those events that form
a maximal $mo$-chain, starting from $e$, of events that either are in
$e$'s thread or are RMWs.\footnote{\citec{5.1.2.4}{10}}

Release/acquire synchronisation is captured by the $sw$ relation. This
relates an atomic write-release event to an atomic read-acquire event
in a different thread if the read obtains its value from the write or
its release sequence.\footnote{\citec{5.1.2.4}{11}} If the acquire
(resp. release) is a fence, the synchronisation happens via an atomic
read (resp. write) sequenced before (resp. after) the
fence.\footnote{\citec{7.17.4}{2--4}}

Happens-before ($hb$) is a transitive relation that includes
sequenced-before and synchronisation edges, and puts initial events
before all other events.\footnote{\citec{5.1.2.4}{18}, simplified in
the absence of "memory\_order\_consume"} We use $hbl$ to abbreviate
happens-before to events on the same location. A write is
\emph{visible} ($vis$) to a read if it is the most recent write to
that location in
happens-before.\footnote{\label{fn:narf}\citec{5.1.2.4}{19}}

Two events are in \emph{conflict} ($cnf$) if they access the same
location and at least one is a write;\footnote{\citec{5.1.2.4}{4}}
these events go on to form a \emph{data race} ($dr$) if they are
unrelated by happens-before, they are not both atomic, and they are in
different threads.\footnote{\citec{5.1.2.4}{25}} 
\end{commentary}

We now use the derived relations of Def.~\ref{def:c11_further_derived} to formalise what it means for an execution to be consistent.

\begin{definition}[Consistency] 
\label{def:consistency_axioms}
A candidate execution $(X,w) = (E,I,lbl,thd,sb,rf,mo,S)$ is
\emph{consistent}, written $\consistent(X,{}$ $w)$, if it is well-formed and
it satisfies all of the following axioms:
{
\renewcommand\jot{1.5mm}
\newcommand\wheregap{34mm}
\begin{gather}
\irreflexive(hb)
\tag{\axiom{Hb}}
\\
\irreflexive((rf^{-1})^? ; mo ; rf^? ; hb) 
\tag{\axiom{Coh}}
\\
\irreflexive(rf ; hb)
\tag{\axiom{Rf}}
\\
\isempty ((rf ; [nal]) \setminus vis) 
\tag{\axiom{NaRf}}
\\
\irreflexive (rf \cup (mo ; mo ; rf^{-1}) \cup (mo ; rf))
\tag{\axiom{Rmw}}
\\
\rlap{$\irreflexive(S ; r_1)$} 
\hspace{\wheregap} \text{where $r_1 = \var{hb}$} 
\tag{\axiom{S1}}
\\
\rlap{$\irreflexive(S ; r_2)$}
\hspace{\wheregap} \text{where $r_2 = \var{Fsb}^? ; \var{mo} ; \var{sbF}^?$} 
\tag{\axiom{S2}}
\\
\rlap{$\irreflexive(S ; r_3)$} 
\hspace{\wheregap} \text{where $r_3 = \var{rf}^{-1} ; [\mosc] ; \var{mo}$} 
\tag{\axiom{S3}}
\\
\rlap{$\irreflexive((S\setminus(mo;S)) ; r_4)$}
\hspace{\wheregap} \text{where $r_4 = \var{rf}^{-1}; \var{hbl} ; [W]$}
\tag{\axiom{S4}}
\\
\rlap{$\irreflexive(S ; r_5)$}
\hspace{\wheregap} \text{where $r_5 = \var{Fsb} ; \var{fr}$} 
\tag{\axiom{S5}}
\\
\rlap{$\irreflexive(S ; r_6)$} 
\hspace{\wheregap} \text{where $r_6 = \var{fr} ; \var{sbF}$} 
\tag{\axiom{S6}}
\\
\rlap{$\irreflexive(S ; r_7)$}
\hspace{\wheregap} \text{where $r_7 = \var{Fsb} ; \var{fr} ; \var{sbF}$} 
\tag{\axiom{S7}}
\end{gather}
}
\end{definition}

\begin{commentary}
These axioms are equivalent to those in Batty et al.'s Lem
formalisation~\cite{batty+12}, the fidelity of which has been endorsed
by the C11 standards committee, but because they are expressed in the
\cat{} language, they are markedly more concise. We have established
this equivalence using the HOL theorem prover, with the help of a tool
we wrote for exporting \cat{} files to Lem, and our proof script is
available online~\cite{overhauling-companion}. We now explain each
axiom in turn.

Happens-before must contain no
cycles.\footnote{\cite[\S1.10:12]{c++11}} Requiring irreflexivity here
is sufficient~(\axiom{Hb}), since $hb$ is transitive.
Coherence~(\axiom{Coh}) governs the relationship between $hb$ and
$mo$: if the write $e_1$ is $mo$-before the write $e_2$, then $e_2$
(and any events that read from $e_2$) must not happen before $e_1$
(nor before any events that read from
$e_1$).\footnote{\citec{5.1.2.4}{7}, \citec{5.1.2.4}{22},
\cite[\S1.10:17--18]{c++11}} A read must not observe a write that
happens after it~(\axiom{Rf}),\footnote{The specification uses the
`visible sequence of side effects' to phrase this
clause~\citec{5.1.2.4}{22}, but Batty~\cite[\S5.3]{batty14} has proved
that `happens after' suffices.} and a read of a non-atomic location
must observe a visible write (\axiom{NaRf}).\footref{fn:narf} An RMW
must observe the immediately-preceding write in
$mo$~(\axiom{Rmw});\footnote{\citec{7.17.3}{12}} that is, not itself
(first disjunct), nor a too-early write (second disjunct), nor a
too-late write (third disjunct).

This leaves the SC axioms, which we present using \emph{where}-clauses
for ease of reference later. Axiom \axiom{S1} states that $S$ must be
consistent with happens-before.\footnote{\label{fn:SCreads}\citec{7.17.3}{6}} Axiom
\axiom{S2} governs the relationship between $S$ and $mo$: if the write
$e_1$ is $mo$-before the write $e_2$, then $e_2$ (and any fences
sequenced after $e_2$) must not come before $e_1$ (nor before
any fences sequenced before $e_1$) in $S$.\footnote{\citec{7.17.3}{6},
\cite[\S29.3:7]{c++11}, \cite[\S29.3:7]{c++14}} 

Axioms \axiom{S3} and \axiom{S4} constrain the values that an SC read
$e_1$ of a location $l$ may observe. If there are any SC writes to $l$
preceding $e_1$ in $S$, then $e_1$ must read either from the most
recent of these in $S$ -- call this $e_2$ -- or from a non-SC write
that does not happen before $e_2$.\footref{fn:SCreads} We
encode this requirement as two irreflexivity constraints. First, we
wish to rule out reading from an SC write that is not the most recent
in $S$; that is, we wish to forbid cycles of the shape depicted below
left, where $S_{\rm loc} \eqdef S\cap {=_{loc}}$. Axiom \axiom{S3} does this, using the simplified form shown
below right.
\begin{center}
\begin{tikzpicture}[inner sep=1pt, baseline=2mm]
\node (a) at (0,0) {$R$};
\node (b) at (1.2,0) {$W$};
\node (c) at (1.2,0.8) {$W$};
\draw[edgeS] (c) to [auto, pos=0.4] node {$S$} (b);
\draw[edgeS] (b) to [auto, swap, pos=0.4] node {$S_{\rm loc}$} (a);
\draw[edgerf] (c) to [auto, bend right, swap] node {$rf$} (a);
\end{tikzpicture}
\begin{tabular}{c}simplifies \\ to\end{tabular}
\begin{tikzpicture}[inner sep=1pt, baseline=2mm]
\node (a) at (0,0) {};
\node (b) at (1,0) {};
\node (c) at (1,0.8) {$\mosc$};
\draw[edgemo] (c) to [auto] node {$mo$} (b);
\draw[edgeS] (b) to [auto, swap] node {$S$} (a);
\draw[edgerf] (c) to [auto, swap, bend right] node {$rf$} (a);
\end{tikzpicture}
\end{center}
Second, we require $e_1$ \emph{not} to read from a write that happens
before $e_2$; that is, we wish to forbid cycles of the shape depicted
below left. Axiom \axiom{S4} does this, using the simplified form
shown below right.
\begin{center}
\begin{tikzpicture}[inner sep=1pt, baseline=2mm]
\node (a) at (0,0) {$R$};
\node (b) at (3.5,0) {$W$};
\node (c) at (3.5,0.8) {$W$};
\draw[-latex] (c) to [auto] node {$hb$} (b);
\draw[edgeS] (b) to [auto, swap, pos=0.47] node {$S_{\rm loc} \setminus (S\mathbin{\semicolon} [W]
\mathbin{\semicolon} S_{\rm loc})$} (a);
\draw[edgerf] (c) to [auto, out=180, in=43, pos=0.8, swap] node {$rf$} (a);
\end{tikzpicture}
\,\,\,
\begin{tabular}{@{}c@{}}simpli- \\ fies to\end{tabular}
\,\,\,
\begin{tikzpicture}[inner sep=1pt, baseline=2mm]
\node (a) at (0,0) {};
\node (b) at (2.2,0) {$W$};
\node (c) at (2.2,0.8) {};
\draw[-latex] (c) to [auto] node {$hbl$} (b);
\draw[edgeS] (b) to [auto, swap,pos=0.45] node {$S \setminus (mo\mathbin{\semicolon}S)$} (a);
\draw[edgerf] (c) to [auto, swap, out=180, in=60, pos=0.8] node {$rf$} (a);
\end{tikzpicture}
\end{center}

Axioms \axiom{S5}, \axiom{S6} and \axiom{S7} govern SC fences. If a
read $e_1$ of a location $l$ is sequenced after an SC fence, then
$e_1$ must not read from a write to $l$ that is $mo$-earlier than the
last write to $l$ that precedes the fence in
$S$.\footnote{\citec{7.17.3}{9}} In fact, `the last write' here can be
safely generalised to `some write', because being $mo$-earlier than
\emph{some} write to $l$ that precedes the fence in $S$ implies being
$mo$-earlier than the last write, since $mo$ is total~(\axiom{S5}). If
a write $e_2$ to location $l$ is sequenced before an SC fence, then
any SC read of $l$ that follows the fence in $S$ must not read from a
write to $l$ that is $mo$-earlier than
$e_2$~(\axiom{S6}).\footnote{\citec{7.17.3}{10}} Finally, if a read
$e_1$ of location $l$ is sequenced after an SC fence, and a write
$e_2$ to $l$ is sequenced before another SC fence that precedes the
first fence in $S$, then $e_1$ must not read from a write $mo$-earlier
than $e_2$~(\axiom{S7}).\footnote{\citec{7.17.3}{11}} \end{commentary}

A final axiom formalises what it means for an execution to exhibit a
fault.
\begin{definition}[Faultiness] A candidate execution $(X,w)$ is
\emph{faulty}, written $\faulty(X,w)$, if it is consistent and
does \emph{not} satisfy the following axiom:
\begin{gather}
\tag{\axiom{Dr}}
\isempty(dr).
\end{gather}
\end{definition}
If any basic execution can be extended to a faulty candidate
execution, then the entire program's behaviour is `undefined' and any
execution is allowed. Otherwise, the allowed executions are those
basic executions that can be extended to a consistent candidate
execution.

\begin{definition}[Allowed executions] Given a set $Xs$ of a program's
basic executions, we obtain the program's \emph{allowed}
executions as:
\begin{center}
$\allowed(Xs) ~~\eqdef~~ \stack{\IF~\exists X \in Xs\ldotp \exists w \ldotp
\faulty(X,w)~\THEN~\mathbb X \\ \ELSE~\{X \in Xs
\mid \exists w\ldotp \consistent(X,w)\}}$
\end{center}
\end{definition}

\section{Overhauling the SC Axioms in C11}
\label{sec:sc}

The rules for SC axioms in C11, as demonstrated in the previous
section, are highly convoluted. In this section, we describe how these
rules can be improved in two fairly orthogonal ways. In
\S\ref{sec:sc-partial}, we describe how the total order over SC
operations can be replaced with a partial order; this simplification
will be demonstrated in \S\ref{sec:experiments} to dramatically improve the
efficiency with which the model can be simulated. In
\S\ref{sec:sc-simp}, we describe a slight strengthening of the model
that enables significant simplifications to be made. These
simplifications lead to a model that is easier to understand, and
should prove easier to work with in a formal setting.

\subsection{Reducing $S$ from a Total to a Partial Order}
\label{sec:sc-partial}

We observe that all but one of the seven SC axioms
(Def.~\ref{def:consistency_axioms}) can be written in the form
$\irreflexive(S ; r)$ for some relational expression $r$. These $r$'s
can be seen as the constraints on the total order $S$. Axiom
\axiom{S4} is not quite of this form. However, replacing its
`$S\setminus(mo;S)$' with just `$S$', to obtain the axiom \axiom{S4a}
given below, happens to coincide exactly with an amendment to the
model already proposed by Vafeiadis et al.\ to lend the model more
desirable mathematical properties~\cite[\S4.2]{vafeiadis+15}.
\begin{gather}
\irreflexive(S ; r_4) \tag{\axiom{S4a}}
\end{gather}
Where axiom \axiom{S4} forbids an SC read to observe any write that
happens before the \emph{most recent} SC write in $S$, axiom
\axiom{S4a} forbids it to observe any write that happens before
\emph{any} SC write in $S$. Let us assume here that the
uncontroversial amendment of Vafeiadis et al. will be accommodated by the C standards
committee.
\begin{lemma}[SC order extension principle]
\label{lem:oep}
For any relation $r$,
there exists a strict total order $S$ over all SC events that is
compatible with $r$, if and only if $r$, when restricted unequal SC events, is acyclic. That is:
\[
(\exists S\ldotp \axiom{WfS} \wedge \irreflexive(S;r)) =
\acyclic(\mosc^2 \setminus id \cap r).
\]
\end{lemma}
\begin{proof}
This follows from the well-known order extension principle: that any
(strict) partial order can be extended to a (strict) total order.
\end{proof}
We are now in a position to replace the seven irreflexivity
axioms with a single acyclicity axiom.
\begin{theorem}
There exists a strict total order on SC events that satisfies axioms
\axiom{S1}, \axiom{S2}, \axiom{S3}, \axiom{S4a}, \axiom{S5},
\axiom{S6}, and \axiom{S7}, if and only if the following \axSpartial{}
axiom (which states that the union of all the constraints on
$S$, when restricted to unequal SC events, is acyclic) holds:
{
\mathindent=2mm
\begin{gather}
\acyclic(\mosc^2 \setminus id \cap (r_1 \cup r_2 \cup r_3 \cup r_4
\cup r_5 \cup r_6 \cup r_7))
\tag{\axSpartial}
\end{gather}
}
That is:
\[
(\exists S\ldotp
\axiom{WfS} \wedge \axiom{S1} \wedge \axiom{S2} \wedge \axiom{S3}
\wedge \axiom{S4a} \wedge \axiom{S5} \wedge \axiom{S6} \wedge
\axiom{S7}) = \axSpartial.
\] 
\end{theorem}
\begin{proof}
\begin{gather*}
\quad \exists S\ldotp
\axiom{WfS} \wedge \axiom{S1} \wedge \axiom{S2} \wedge \axiom{S3}
\wedge \axiom{S4a} \wedge \axiom{S5} \wedge \axiom{S6} \wedge
\axiom{S7}
\\ 
= \text{[basic properties of relations]}
\\
\quad \exists S\ldotp \axiom{WfS} \wedge \irreflexive(S; (r_1 \cup r_2 \cup r_3 \cup
r_4 \cup r_5 \cup r_6 \cup r_7)) 
\\
= \text{[by Lemma~\ref{lem:oep} with $r$ instantiated to $r_1 \cup
\dots \cup r_7$]}
\\
\quad \acyclic(\mosc^2 \setminus id \cap (r_1 \cup r_2 \cup r_3 \cup
r_4 \cup r_5 \cup r_6 \cup r_7)) \hfill\qedhere
\end{gather*}
\end{proof}
Having replaced axioms \axiom{S1}--\axiom{S7} with the new
\axSpartial{} axiom, we no longer require the $S$ relation in
execution witnesses. Memory model simulators, such as \herd{},
typically work by enumerating all executions of a program and then
filtering out the consistent subset. Removing the need to iterate
through all possible total orders of SC events -- a computation that
is exponential in the number of SC events -- allows simulation
performance to be greatly improved, as demonstrated in
\S\ref{sec:experiments}.

\subsection{A Stronger and Simpler SC Axiom}
\label{sec:sc-simp}

We now show that it is possible to strengthen the SC semantics without
requiring changes to the compilation schemes of any of the C11 target
architectures that have an established formal memory model, that is:
x86 and Power. The strengthening we propose simplifies the
\axSpartial{} axiom significantly and provides stronger guarantees to
the programmer.

The proposal for this simplification arises from the observation that
the relations considered in the \axSpartial{} axiom are nearly
symmetric in $hb$, $mo$ and $fr$. In particular, both $hb$ and $mo$
constrain the $S$ order between any combination of SC fences and
atomics. The treatment of $fr$ is different: for $fr$ edges that begin
or end at a fence, the axioms \axiom{S5}, \axiom{S6} and \axiom{S7}
ensure that the SC order is constrained to match. When two SC atomics
are related by an $fr$ edge (\axiom{S3} and \axiom{S4}), ordering is
only provided when the intermediate access that forms the $fr$ is
itself an SC atomic (rule \axiom{S3}), or when the $mo$ edge from the
intermediate access of the $fr$ to its target is also covered by a
$hb$ edge (rule \axiom{S4a}).

Our proposal is to strengthen the \axSpartial{} axiom, to add these
missing constraints so that every $fr$ edge between SC atomics
contributes to the $S$ order. We achieve this in our model by removing
the $[\mosc]$ restriction from \axiom{S3}, which results in the
following axiom:
\begin{align}
& \irreflexive(S ; fr). \tag{\axiom{S3a}} 
\end{align}
This change permits a significant simplification to the SC rules that
we establish in the following theorem.
\begin{theorem}
\label{thm:SC_simp}
If rule \axiom{S3} is replaced by \axiom{S3a} (that is, if $r_3$ is
replaced with $fr$ in the \axSpartial{} axiom) then \axSpartial{}
becomes equivalent to:
\begin{gather}
\tag{\axSsimp}
\acyclic(\mosc^2 \setminus id \cap (Fsb^?; (hb \cup fr \cup mo) ; sbF^?)).
\end{gather}
That is:
\[
\acyclic(\mosc^2 \setminus \mathit{id} \cap (r_1 \cup r_2 \cup \mathit{fr} \cup r_4
\cup r_5 \cup r_6 \cup r_7)) = \axSsimp.
\] 

\end{theorem}
\begin{proof}
\begin{gather*}
\quad r_1 \cup r_2 \cup fr \cup r_4 \cup r_5 \cup r_6 \cup r_7 
\\ 
= \text{[unfolding definitions and combining $\mathit{fr}$, $r_5$, $r_6$ and $r_7$]}
\\
\quad hb \cup (Fsb^? ; mo ; sbF^?) \cup (Fsb^? ; fr ; sbF^?)
\cup r_4 
\\
= \text{[since $r_4 \subseteq \mathit{fr}$, by \axiom{WfMo}]}
\\
\quad hb \cup (Fsb^? ; mo ; sbF^?) \cup (Fsb^? ; fr ; sbF^?) 
\\
= \text{[since $\mathit{hb} = (\mathit{Fsb}^? ; \mathit{hb} ; \mathit{sbF}^?)$]} 
\\
\quad Fsb^? ; (hb \cup fr \cup mo) ; sbF^?\hfill\qedhere
\end{gather*}
%
\end{proof}
%


\paragraph{Programming impact}
The change presented here does strengthen the memory model; there are
executions that were previously allowed that are now forbidden. The
simplest we found, which is similar to one used by Vafeiadis et
al.~\cite[Fig.~6]{vafeiadis+15}, is presented in
Example~\ref{ex:c11_candidate_exec}. We believe
Example~\ref{ex:c11_candidate_exec} to be a counterintuitive
execution, because the read event $i$ does not observe the most recent
write to "x" in $S$ (namely, $f$), but $c$, which is $mo$-earlier than
$f$. The execution is forbidden by axiom \axiom{S3a} because of its
$f\xrightarrow{S}i\xrightarrow{fr}f$ cycle. Although the current C11
model allows this execution, mapping this example to the formalised
targets of C11 (Power and x86) never yields programs that exhibit it.

\subsection{Soundness of Existing C11 Compilation Schemes}
\label{sec:c11_soundness}

There are two C11 targets with formal architectural memory models: x86
and Power. In this subsection, we establish that for both of these
architectures, the strengthening does not require a stronger
compilation mapping. In both cases, we rely on an existing proof of
soundness from the literature. We need only establish that our
strengthened \axSsimp{} axiom holds.

To establish the soundness of our strengthening for x86, we build on
the soundness proof of Batty et al.~\cite{batty+11}, which uses the
axiomatic model of x86 of Owens et al.~\cite{owens+09}. To obtain
soundness for Power, we build on the soundness proof of Batty et
al.~\cite{batty+12}, which uses the operational Power model of Sarkar
et al.~\cite{sarkar+11}.


\newcommand\soundnesslemma{
Let $P$ be a C11 program that has no faulty executions. 
If we compile $P$ to x86 according to the mapping given by Batty et
al.~\cite{batty+11}, then every valid x86 execution
corresponds to a C11 execution where $\axSsimp{}$
holds. If we compile $P$ to Power according to the mapping given by
Batty et al.~\cite{batty+12}, then every valid Power trace is
observationally equivalent to a C11 execution where $\axSsimp{}$ holds.
%
%
}

\begin{theorem}
\label{lem:soundnesslemma}
\soundnesslemma \hfill [Proof in \S\ref{appx:proofs}]
\end{theorem}

\begin{remark}[Soundness of the ARMv8 compilation scheme]
At the time of writing, work to formalise the ARMv8 specification, and
how it implements C11, is ongoing~\cite{flur+16}. We understand that
it is not currently clear whether the specification is intended to
allow or forbid behaviours like our
Example~\ref{ex:c11_candidate_exec}, and whether the effects of this
decision on the C11 memory model are understood. As such, we see our
work as a timely intervention in the ongoing argument about how this
particular aspect of the ARMv8 specification should evolve and be
formalised.
\end{remark}

\subsection{Effect on the Standard}
\label{sec:c11_proposal} 
We give below a suggestion for how the wording of the standard could
be changed to accommodate our proposal. Our text, which replaces
paragraphs 6 and 9--11 of section 7.17.3, is considerably shorter
(\CnewWORDS{} words rather than \ColdWORDS{}) while preserving the
style and terminology of the original. We have retained the total
order $S$ in our wording, because we believe it is more intuitive for
programmers than an acyclicity condition. Nonetheless, we enable
efficient simulation of this model via the \axSsimp{}
axiom (which is equivalent to the total order formulation, thanks to Lemma~\ref{lem:oep} with $r$ instantiated to $Fsb^? ; (hb \cup fr \cup mo) ; sbF^?$).
\begin{oframed}
\topsep=0pt
\begin{enumerate}
\item A value computation $A$ of an object $M$ \emph{reads before} a side
effect $B$ on $M$ if $B$
follows, in the modification order of $M$, the side effect that $A$
observes. 
\item If $X$ reads before $Y$, or happens before $Y$, or precedes $Y$
in modification order, then $X$ (and any fences sequenced
before $X$) is \emph{SC-before} $Y$ (and any fences sequenced
after $Y$).
\item There shall be a single total order $S$ on all {\tt memory\_}
{\tt order\_seq\_cst} operations, consistent with the SC-before
order.\end{enumerate}
\end{oframed}

\paragraph{Summary}
This section has described how, having strengthened the
original set of axioms (\axiom{S1} through \axiom{S7}) to use
Vafeiadis et al.'s \axiom{S4a} in place of \axiom{S4}, the behaviour
of SC operations can be captured by a single axiom (\axSpartial) that
allows the total order $S$ to be eliminated from the model. Moreover,
if the axioms are further strengthened to use our \axiom{S3a} in place
of \axiom{S3}, then that axiom can be greatly simplified (\axSsimp),
while still respecting current compilation schemes.

\section{Formalising the OpenCL Memory Model}
\label{sec:opencl.cat}

A principal aim of the OpenCL initiative is to provide functional
portability across a plethora of heterogenous many-core devices.  The
standard is implemented by CPU, GPU and FPGA vendors, and aims to
allow applications to be device-agnostic.  The OpenCL memory model,
introduced in the 2.0 revision of the standard, is inherited from that
of C11, but is specialised and extended for heterogeneous programming.
The memory model is the sole mechanism for correctly implementing
fine-grained concurrent algorithms in a device-agnostic manner.
Rigorous foundations for this model are thus vital.

We now describe how our formalisation of the C11 memory model
(\S\ref{sec:c11.cat}, \S\ref{sec:sc}) can be extended to yield the first mechanised
formalisation of the full OpenCL memory model.
We describe the form of OpenCL programs (\S\ref{sec:opencl_programs}),
their executions (\S\ref{sec:opencl_executions}), and the
axioms against which these executions are judged
(\S\ref{sec:opencl_axioms}). We then discuss some interesting
features of the memory model: some innocuous quirks
(\S\ref{sec:opencl_quirks}) and some serious shortcomings
(\S\ref{sec:opencl_problems}). The most serious shortcoming relates to
the axioms that govern SC atomics, and we propose how to fix this in \S\ref{sec:opencl_sc}.

For reasons of space, and because they are orthogonal to the thrust of
our contributions, we omit our treatment of barrier synchronisation
operations and the associated issue of \emph{barrier divergence}.  As
with the omitted C11 features mentioned in \S\ref{sec:c11.cat}, our
\cat{}-based formalisation of the OpenCL memory model, provided on our
companion webpage~\cite{overhauling-companion}, fully
accounts for these features.

\subsection{OpenCL Programs}
\label{sec:opencl_programs}

\begin{definition}[Structure of OpenCL programs]
\label{def:opencl_syntax}
Building on Def.~\ref{def:c11_syntax}, we consider OpenCL programs of
the form
\begin{center}
$P = \bigquadparallel_{d\in D} \bigtripleparallel_{w\in W}
\bigparallel_{t\in T} p_{d,w,t}$
\end{center}
where $D$, $W$, and $T$ are sets of device, work-group, and thread
identifiers, and each $p_{d,w,t}$ is a piece of
sequential code. 
\end{definition}

Using the notation above, we can write $p\quadparallel p'$ to denote a
litmus test comprising two threads to be executed on different
devices, $p\tripleparallel p'$ for two threads in different
work-groups in the same device, and $p\parallel p'$ for two threads in
the same work-group. We can also write, for example,
$p_1 \parallel p_2 \tripleparallel p_3 \parallel p_4 \quadparallel p_5
\parallel p_6 \tripleparallel p_7 \parallel p_8$,
to denote a litmus test comprising two devices, each executing two
work-groups, each containing two threads.

\begin{remark}[Limitations]
This program structure does not account for \emph{sub-groups}, an
optional extension in OpenCL 2.0 that allows threads to synchronise
with one another at a level of granularity finer than that of a
work-group,\footnote{Sub-groups have become a core feature in the
recent OpenCL~2.1 specification~\cite[p.~22]{opencl21}.} nor for
further non-OpenCL threads (e.g., POSIX threads) running on the host
platform.
Moreover, $w$ and $t$ can actually be 1-, 2-, or 3-dimensional
vectors, but we make the simplifying assumption that all identifiers
are natural numbers.
\end{remark}

Recall that locations in C11 are either $\nonatomic$ or $\atomic$
(Def.~\ref{def:c11_locs}). OpenCL locations are further declared to reside in a \emph{memory region}.

\begin{definition}[Memory regions]
\label{def:opencl_regions}
We have $region(l) \in \{"local",$ $"global", "global\_fgb"\}$ for every
location $l$, where "fgb" stands for
\underline{\smash{f}}ine-\underline{\smash{g}}rained shared virtual memory (SVM)
\underline{\smash{b}}uffer.\footnote{OpenCL also provides "private" regions,
each accessible only to one thread, and a read-only "constant" region,
but neither of these are interesting from a memory modelling
perspective.} There is one "local" region per work-group, containing
locations accessible only to that work-group. Locations in the
"global" or "global\_fgb" region are accessible to all devices. Fences
can be performed either on the global memories ("global" and
"global\_fgb") or on the local memory, or both simultaneously.
\end{definition}

The distinction between "global" and "global\_fgb" locations is that
the former must not be shared between different devices, while the
latter enable inter-device communication. Unlike C11, in which any
memory location can be shared between threads, the OpenCL memory model
physically prevents certain sharing patterns. For instance, threads
from different devices are \emph{forbidden} from conflicting on
"global" memory, but are \emph{able} to do so as a result of a
programmer fault; in contrast, threads from different work-groups are
\emph{unable} to conflict on local memory: the language provides no
mechanism through which such a conflict can arise.
%
%
%
\begin{definition}[Memory scopes]
\label{def:opencl_scopes}
Atomics in OpenCL are parameterised by a \emph{memory scope}. The
three options are
\[
\begin{array}{r@{~}ll}
s ::= & \swg & \text{(work-group scope)} \\
 \mid & \sdv & \text{(device scope)} \\
 \mid & \sall & \text{(system scope).}
\end{array}
\]
A memory scope specifies how widely visible the effects of the
operation should be.
\end{definition}
\begin{Example}
\label{ex:openclmp}
The use of memory scopes is illustrated by the
following code, which implements the message-passing idiom between two
threads in the same work-group.
\begin{center}
\vspace*{-1mm}
\begin{tabular}{l@{~~}||@{~~}l}
\multicolumn{2}{c}{\texttt{global int *x; global atomic\_int *y;}} \\
\texttt{*x = 42;} & \texttt{if(load(y,$\moacq$,$\swg$)==1)} \\
\texttt{store(y,1,$\morel$,$\swg$);} & \texttt{~~r = *x;}
\end{tabular}
\vspace*{-1mm}
\end{center}
Since all accesses to the global location {\tt y} come from the same
work-group, those accesses can be performed at $\swg$ scope (which
means that on implementations where each work-group caches global
memory, it suffices to read/write those cached values). This scope
would be insufficient, and the program deemed faulty, if the threads
were in different work-groups -- both scopes would have to be upgraded
to $\sdv$.
\end{Example}

\subsection{OpenCL Executions}
\label{sec:opencl_executions}

OpenCL executions extend C11 executions as follows.

\begin{definition}[OpenCL event labels]
We extend C11 event labels (Def.~\ref{def:c11_event_labels}) with an
additional $scope$ attribute, which assigns a memory scope $s$ to all
atomic events. We also subdivide the $\evF$ label in order to
represent fences on global ($\evFG$), local ($\evFL$) and both-global-and-local
memory ($\evFGL$). The updated table is as follows:
\[
\begin{array}{l@{~}l@{}ccccc@{}l|cccc}
kind & & loc & rval & wval & ord & 
scope & & R & W & F & A \\ \hline
\evWna &(& l, & & v, & & 
&)& & \checkmark & & \\
\evW &(& l, & & v, & o, & 
s &)& & \checkmark & & \checkmark \\
\evRna &(& l, & v, & & & 
&)& \checkmark & & & \\ 
\evR &(& l, & v, & & o, & 
s &)& \checkmark & & & \checkmark\\ 
\evRMW &(& l, & v, & v', & o, & 
s &)& \checkmark & \checkmark & & \checkmark\\ 
\evFG &(& & & & o, & s &)& & & \checkmark & \checkmark \\
\evFL &(& & & & o, & s &)& & & \checkmark & \checkmark \\
\evFGL &(& & & & o, & s &)& & & \checkmark & \checkmark 
\end{array}
\]

\end{definition}
\begin{definition}[OpenCL executions] 
\label{def:opencl_executions}
An OpenCL execution is a tuple $(E, I, lbl, thd, wg, dv, sb)$ where
$(E, I, lbl, thd, sb)$ is a C11 execution as in
Def.~\ref{def:c11_executions}, and $wg,dv\subseteq(E\setminus I)^2$
are equivalence relations on non-initial events that relate events
from the same work-group and device, respectively. In order to enforce
the privacy of "local" locations to a single work-group, we require
that if $loc(e) = loc(e') = l$ and $region(l) = "local"$, then
$(e,e')\in wg$. 
\end{definition}
\begin{definition}[Derived sets and relations]
In the context of an OpenCL execution $(E,I,lbl,thd,$ $wg,dv,sb)$, we
define
\begin{align*}
fgb &~\eqdef~ \{e\in E\setminus F \mid region(loc(e)) = "global\_fgb"\} \\
G &~\eqdef~ \stack{\{e \in F \mid kind(e) \in \{\evFG, \evFGL\}\} \cup {}\\ \{e\in E\setminus F \mid
region(loc(e)) = "global"\} \cup fgb} \\
L &~\eqdef~ \stack{\{e \in
F \mid kind(e) \in \{\evFL, \evFGL\}\}
\cup {}\\ \{e\in E\setminus F \mid region(loc(e)) = "local"\} }
\end{align*}
as the sets of events that access, respectively: fine-grained atomic
SVM buffers, global memory, and local memory. Also, for each scope $s$, we abbreviate the set $\{e \in A \mid scope(e) = s\}$ as
just $s$.
\end{definition}

\begin{definition}[OpenCL candidate executions] Candidate executions in
OpenCL, and their well-formedness, are defined in the same way as in
C11 (Def.~\ref{def:candidate_exec}).
\end{definition}

\subsection{OpenCL Axioms}
\label{sec:opencl_axioms}

We now define and discuss the $\consistent$ and $\faulty$ predicates
for the OpenCL memory model, paying particular attention to each of
the departures from C11. We justify our formal definitions by
reference to the OpenCL specification~\cite{opencl21}, writing $n$/$m$
to denote line $m$ on page $n$.

\begin{definition}[Further derived sets and relations]
\label{def:opencl_derived}
In the context of a candidate execution $(E,I,lbl,thd,sb,rf,mo,S)$, we
define the following subsets of $E$ and relations over $E$:
\begin{eqnarray*}
incl &\eqdef& (\swg^2 \cap wg) \cup (\sdv^2 \cap dv) \cup \sall^2
\\
rsw(r) &\eqdef& \stack{([r \cap rel] ; Fsb^?
; [W \cap A] ; rs^? ; [r] ; rf ; {}\\{} [R \cap A] ;  sbF^?; [r \cap acq]) \cap incl \setminus
thd}
\\
gsw &\eqdef& rsw(G) \cup (rsw(L) \cap (\mosc^2 \cup (G \cap L \cap F)^2))
\\
lsw &\eqdef& rsw(L) \cup (rsw(G) \cap (\mosc^2 \cup (G \cap L \cap F)^2))
\\
ghb &\eqdef& (G^2 \cap (sb \cup (I \times \neg I)) \cup
gsw)^+ 
\\
lhb &\eqdef& (L^2 \cap (sb \cup (I \times \neg I)) \cup
lsw)^+ 
\\
ghbl &\eqdef& ghb \cap {=_{loc}} \\
lhbl &\eqdef& lhb \cap {=_{loc}} \\
gvis &\eqdef& (W \times R) \cap ghbl \setminus(ghbl ; [W] ; ghb) 
\\
lvis &\eqdef& (W \times R) \cap lhbl \setminus(lhbl ; [W] ; lhb)
\\
hr &\eqdef& \stack{cnf \setminus (ghb \cup lhb) \setminus (ghb \cup lhb)^{-1} 
\setminus incl \setminus thd}
\\
iddr &\eqdef& cnf \setminus dv \setminus fgb^2
\\
\sccondI &\eqdef& \neg(E^2 ; [\mosc \setminus (\sall \cap fgb)] ; E^2)
\\
\sccondII &\eqdef& \neg(E^2 ; [\mosc \setminus (\sdv \setminus
fgb)] ; E^2)
\end{eqnarray*}
\end{definition}
\begin{commentary}
In OpenCL, only events that have \emph{inclusive} scopes ($incl$) can
synchronise: either the events have $\swg$ scope and are in the same
work-group, or they have $\sdv$ scope and are in the same device, or
they have $\sall$ scope.\footnote{\citecl{47}{16}{47}{26}} We shall
explain in \S\ref{sec:opencl_problems} how this notion of scope
inclusion is unnecessarily conservative.

The synchronisation relation ($rsw$) is parameterised by a region $r$
(global or local). The global synchronises-with relation ($gsw$)
includes events that synchronise on global
memory,\footnote{\citecl{51}{1}{51}{9}} but also includes events that
synchronise on \emph{local} memory, providing both events have memory
order $\mosc$,\footnote{\citecl{51}{32}{51}{33}} or both are
global-and-local fences.\footnote{\citecl{54}{13}{54}{16}} Local
synchronises-with ($lsw$) is analogous. 
Example~\ref{ex:fence_gl} shows how synchronisation works in
the presence of global-and-local fences.

Happens-before is partitioned into global and local versions: global
happens-before ($ghb$) contains global synchronises-with and
sequenced-before edges between events on global
memory,\footnote{\citecl{49}{3}{49}{7}} and local happens-before
($lhb$) is analogous.\footnote{\citecl{49}{8}{49}{11}} See
Example~\ref{ex:global-local-mp} for a discussion of the repercussions
of this definition of happens-before. Visibility is also split into
global ($gvis$) and local ($lvis$)
versions.\footnote{\citecl{49}{21}{49}{26}}

The \emph{heterogeneous race} ($hr$)\footnote{This terminology is due
to Hower et al.~\cite{hower+14}.} generalises C11's data race ($dr$,
Def.~\ref{def:c11_further_derived}), to reflect the fact that in
OpenCL, even atomic operations can race when memory scopes are used
incorrectly.\footnote{\citecl{49}{29}{49}{33}} If two events from
different devices conflict on a location that is not in a fine-grained
atomic SVM buffer, then they form an \emph{inter-device data race}
($iddr$); such races cannot be ruled out by happens-before
edges.\footnote{\citecl{58}{24}{58}{27}}

This leaves the $\sccondI$ and $\sccondII$ relations. In
OpenCL, the total order $S$ is only required to exist when
\begin{gather*}
\mosc \subseteq \sall \cap fgb \quad\text{or}\quad \mosc \subseteq
\sdv \setminus fgb.
\end{gather*}
The first condition holds when every SC event has $\sall$ scope and
accesses a "global\_fgb" location;\footnote{\citecl{51}{15}{51}{17}}
the second holds when every SC event has $\sdv$ scope and does
\emph{not} access a "global\_fgb"
location.\footnote{\citecl{51}{18}{51}{20}} The relation
$\sccondI$ (resp.~$\sccondII$) is the universal relation if the
first (resp.~$\sccondII$) condition holds and is the empty relation
otherwise. In \S\ref{sec:opencl_problems}, we shall criticise
these conditions as being simultaneously too strong for programmers and
too weak for compiler-writers.
\end{commentary}

\begin{definition}[Consistency axioms in OpenCL]
\label{def:opencl_consistency_axioms}
There are nine consistency axioms. Departures from the C11 consistency
axioms (Def.~\ref{def:consistency_axioms}) are highlighted.
{
\renewcommand\jot{1.5mm}
\begin{gather}
\tag{\oaxiom{HbG}}
\irreflexive(\mhl{ghb})
\\
\tag{\oaxiom{HbL}}
\irreflexive(\mhl{lhb})
\\
\tag{\oaxiom{CohG}}
\irreflexive((rf^{-1})^? ; mo ; rf^? ; \mhl{ghb})
\\
\tag{\oaxiom{CohL}}
\irreflexive((rf^{-1})^? ; mo ; rf^? ; \mhl{lhb})
\\
\tag{\oaxiom{Rf}}
\irreflexive(rf ; \mhl{(ghb\cup lhb)})
\\
\tag{\oaxiom{NaRfG}}
\isempty ((rf ; [\mhl{G \cap {}} nal]) \setminus \mhl{gvis})
\\
\tag{\oaxiom{NaRfL}}
\isempty ((rf ; [\mhl{L \cap {}} nal]) \setminus \mhl{lvis})
\\
\tag{\oaxiom{Rmw}}
\irreflexive (rf \cup (mo;mo;rf^{-1}) \cup (mo;rf))
\\
\tag{\axOSsimp}
\acyclic(\stack{\mosc^2 \setminus id \cap \mhl{(\sccondI \cup \sccondII)} \cap {}\\
(Fsb^? ; (\mhl{ghb \cup lhb} \cup fr \cup mo) ; sbF^?))}
\end{gather}
}
\end{definition}

\begin{commentary} 
Both happens-before relations are required to be
acyclic~(\oaxiom{HbG},
\oaxiom{HbL}).\footnote{\citecl{49}{12}{49}{13}} OpenCL requires
coherence for both global and local happens before
separately~(\oaxiom{CohG},
\oaxiom{CohL}).\footnote{\citecl{50}{11}{50}{24}} The axioms governing
the reads-from relation are carried over from C11~(\oaxiom{Rf},
\oaxiom{NaRfG}, \oaxiom{NaRfL}, \oaxiom{Rmw}), but appropriately
divided into global and local
versions.\footnote{\citecl{49}{26}{49}{27}, \citecl{50}{8}{50}{9},
\citecl{52}{22}{52}{23}}

OpenCL defines the same SC axioms that we saw in
Def.~\ref{def:consistency_axioms} (\axiom{S1}--\axiom{S7}), but uses
$ghb\cup lhb$ in place of $hb$. We have incorporated into axiom
\axOSsimp{} the simplifications that we already discussed in the
context of C11 (\S\ref{sec:sc}). Intersecting with the $\sccondI$
and $\sccondII$ conditions means that the acyclicity constraint is
only enforced when one of those conditions holds.\footnote{\citecl{51}{14}{51}{14}}
\end{commentary}

\begin{definition}[Faultiness in OpenCL]
A candidate OpenCL execution is faulty if it is consistent and does
\emph{not} satisfy both of the following axioms:
{
\renewcommand\jot{1.5mm}
\begin{gather}
\tag{\oaxiom{Hr}} 
\isempty(hr)
\\
\tag{\oaxiom{Iddr}} 
\isempty(iddr)
\end{gather}
}
\end{definition}

\subsection{Quirks in the Memory Model}
\label{sec:opencl_quirks}

We present three worked examples that illustrate features of the
memory model that may not be obvious from a cursory glance at its
axioms. These `quirks' in the model are distinguished from the
technical shortcomings that we save for \S\ref{sec:opencl_problems}.

Our first example illustrates an interesting consequence of
OpenCL's separation of happens-before into two distinct relations.

\begin{Example}
\label{ex:global-local-mp}
Suppose the code of Example~\ref{ex:openclmp} were changed so that
{\tt y} were declared {\tt local} rather than {\tt global}. Executions
such as the one below would then become consistent, which means that
a stale value of "x" can be read (event $f$), even when successful
release/acquire synchronisation (between $d$ and $e$) has occurred.
\begin{center}
\begin{tikzpicture}[inner sep=1pt]
\node[event, anchor=west](a) at (0.5,1.8) 
{\evtlbl{$a$}$\evWna({\tt x},0)$};

\node[event, anchor=west](b) at (2.5,1.8) 
{\evtlbl{$b$}$\evWna({\tt y},0)$};

\node[event, anchor=west](c) at (0,1) 
{\evtlbl{$c$}$\evWna({\tt x},42)$};

\node[event, anchor=west](d) at (0,0.2) 
{\evtlbl{$d$}$\evW({\tt y},1,\morel,\swg)$};

\node[event, anchor=west](e) at (3,1) 
{\evtlbl{$e$}$\evR({\tt y},1,\moacq,\swg)$};

\node[event, anchor=west](f) at (3,0.2) 
{\evtlbl{$f$}$\evRna({\tt x},0)$};

\draw[edgesb] (c.south) to[auto,pos=0.4]
node{$sb$} (d.north -| c.south);

\draw[edgesb] (e.south) to[auto,swap,pos=0.4]
node{$sb$} (f.north -| e.south);

\draw[edgemo] (a) to[auto,pos=0.7,swap] node{$mo$} (c);

\draw[edgemo] (b) to[auto,pos=0.1]
node{$mo$} (d);

\draw[edgerf] (d) to[auto,pos=0.4]
node{$rf$} (e);

\draw[edgerf] (a) to[auto,pos=0.2]
node{$rf$} (f);
\end{tikzpicture}
\end{center}
This execution is consistent because the $sb$ edges no longer induce
either variety of happens-before, since they link events that act on
different memory regions. Worse still, there is now a data race
between $c$ and $f$, which renders the entire program undefined.
\end{Example}

We learn from Example~\ref{ex:global-local-mp} that a flag in one
memory region cannot be used to protect data in another region. To
address this issue, OpenCL provides fences that act on both global
\emph{and} local memory simultaneously. These are illustrated in Example~\ref{ex:fence_gl}.

\begin{Example} 
\label{ex:fence_gl}
The following program uses relaxed ($\morlx$) accesses on the local flag {\tt
y}, relying instead on the fences to synchronise the threads and
enable the global data {\tt x} to be passed.
\begin{center}
\begin{tabular}{@{}l@{~}||@{~}l@{}}
\multicolumn{2}{c}{\texttt{global int *x; local atomic\_int *y;}} \\
\texttt{*x = 42;} & \texttt{if(load(y,$\morlx$,$\swg$)==1)} \\
\texttt{fence(GL,$\morel$,$\swg$);} & \texttt{\{~fence(GL,$\moacq$,$\swg$);} \\
\texttt{store(y,1,$\morlx$,$\swg$);} &
\texttt{~~r = *x;~\} } 
\end{tabular}
\end{center}
The "fence" instructions successfully prevent the stale value of "x"
being read, because the following execution is inconsistent.
\begin{center}
\begin{tikzpicture}[baseline=10mm,inner sep=1pt]
\node[event, anchor=west](a) at (0.5,2.15) 
{\evtlbl{$a$}$\evWna({\tt x},0)$};

\node[event, anchor=west](b) at (2.4,2.15) 
{\evtlbl{$b$}$\evWna({\tt y},0)$};

\node[event, anchor=west](c) at (0,1.5) 
{\evtlbl{$c$}$\evWna({\tt x},1)$};

\node[event, anchor=west](d) at (0,0.75)
{\evtlbl{$d$}$\evFGL(\morel,\swg)$};

\node[event, anchor=west](e) at (0,0) 
{\evtlbl{$e$}$\evW({\tt y},1,\morlx,\swg)$};
\node[event, anchor=west](f) at (3,1.5) 
{\evtlbl{$f$}$\evR({\tt y},1,\morlx,\swg)$};

\node[event, anchor=west](g) at (3,0.75)
{\evtlbl{$g$}$\evFGL(\moacq,\swg)$};

\node[event, anchor=west](h) at (3,0) 
{\evtlbl{$h$}$\evRna({\tt x},0)$};

\draw[edgesb] (c) to[auto,swap,pos=0.4] 
node{$sb$} (d.north -| c.south);

\draw[edgesb] (d.south -| c.south) to[auto,swap]
node{$sb$} (e.north -| c.south);
\draw[edgesb] (f) to[auto,swap,pos=0.4] 
node{$sb$} (g.north -| f.south);

\draw[edgesb] (g.south -| f.south) to[auto,swap] 
node{$sb$} (h.north -| f.south);

\draw[edgerf] ([xshift=-3mm]e.north east) to[auto,pos=0.15, swap, bend
left=10] 
node{$rf$} (f);

\draw[edgerf] (a) to[auto,pos=0.3, swap, bend right=10] 
node{$rf$} (h);

\draw[edgemo] (a) to[auto, pos=0.7, swap] 
node{$mo$} (c);

\draw[edgemo] (b) to[auto, pos=0.3, swap] 
node{$mo$} ([xshift=-4mm]e.north east);
\end{tikzpicture}
\end{center}
The execution is inconsistent because it has a cycle
$h \xrightarrow{rf^{-1}} a \xrightarrow{mo} c \xrightarrow{ghb}
h$, in violation of \oaxiom{CohG}.
Note that $c \xrightarrow{ghb} h$ holds here because, firstly, $(d,g)$ is
in $\var{rsw}(L)$ and hence in $gsw$ and $ghb$, and secondly, $(c,d)$
and $(g,h)$ are both in $sb \cap G^2$ and hence in $ghb$.
\end{Example}

In Example~\ref{ex:work_item_scope}, we illuminate the relationship
between memory scopes and non-atomic operations. Since scopes can be
used to limit atomic operations to certain groups of threads, it is
tempting to introduce an additional `work-item' scope, $\swi$, and
encode non-atomic events as atomic events whose scope is limited to
the current thread. This would make the $\evWna$ and $\evRna$ labels
redundant. An ordinary data race can then be cast as a failure of
scope inclusion. However, the differences between non-atomic and
atomic operations go beyond racy behaviours, as we shall see in the
following example.

\begin{Example}
\label{ex:work_item_scope}
Consider the following load-buffering litmus test:
\begin{center}
\begin{tabular}{l@{~~}||@{~~}l}
\multicolumn{2}{c}{\texttt{global int *x, *y;}} \\
\texttt{if (*x==1) *y=1;} & \texttt{if (*y==1) *x=1;}
\end{tabular}
\end{center}
The execution of this program that exhibits the relaxed behaviour, in
which both comparisons succeed (shown below left), is \emph{not
consistent}: both of its reads observe writes that are not visible, in
violation of the \axiom{NaRf} axiom. If the non-atomic locations
become atomic and the non-atomic operations become work-item-scoped
atomics, then this relaxed behaviour (shown below right) becomes
\emph{consistent}, since the \axiom{NaRf} restriction no longer
applies.
\begin{center}
\begin{tikzpicture}[inner sep=1pt]
\node[event, anchor=west](c) at (0,1) 
{\scriptsize $\evRna("x",1)$};

\node[event, anchor=west](d) at (0,0) 
{\scriptsize $\evWna("y",1)$};

\node[event, anchor=west](e) at (2.3,1) 
{\scriptsize $\evRna("y",1)$};

\node[event, anchor=west](f) at (2.3,0) 
{\scriptsize $\evWna("x",1)$};

\draw[edgesb] (c.south) to[auto,pos=0.4]
node{$sb$} (d.north -| c.south);

\draw[edgesb] (e.south) to[auto,swap,pos=0.4]
node{$sb$} (f.north -| e.south);

\draw[edgerf] (d) to[auto,pos=0.3,swap] node{$rf$} (e);

\draw[edgerf] (f) to[auto,pos=0.7,swap] node{$rf$} (c);
\end{tikzpicture}
\quad
\begin{tikzpicture}[inner sep=1pt]
\node[event, anchor=west](c) at (0,1) 
{\scriptsize $\evR("x",1,\morlx,\swi)$};

\node[event, anchor=west](d) at (0,0) 
{\scriptsize $\evW("y",1,\morlx,\swi)$};

\node[event, anchor=west](e) at (2.3,1) 
{\scriptsize $\evR("y",1,\morlx,\swi)$};

\node[event, anchor=west](f) at (2.3,0) 
{\scriptsize $\evW("x",1,\morlx,\swi)$};

\draw[edgesb] (c.south) to[auto,pos=0.4]
node{$sb$} (d.north -| c.south);

\draw[edgesb] (e.south) to[auto,swap,pos=0.4]
node{$sb$} (f.north -| e.south);

\draw[edgerf] (d) to[auto,pos=0.3,swap] node{$rf$} (e);

\draw[edgerf] (f) to[auto,pos=0.7,swap] node{$rf$} (c);
\end{tikzpicture}
\end{center}
\end{Example}

\subsection{Problems with the Memory Model}
\label{sec:opencl_problems}

We present three shortcomings in the OpenCL memory model, which we
discovered as a direct result of our formalisation efforts.

\paragraph{Scope inclusion is too strong}

The specification provides an overly conservative notion of scope
inclusion: two events only have inclusive scopes if their 
scopes match exactly. This leads to such surprises as the following example.

\begin{Example}
\label{ex:asymmetric_scopes}
Suppose the code of Example~\ref{ex:openclmp} were changed so that
the store to {\tt y} now occurs at $\sdv$ scope, but the load of {\tt
y} remains at $\swg$ scope. Although the release scope is clearly `wide
enough', it does not match the acquiring scope, so no synchronisation
edge is induced. This leads to two data races: both between the non-atomic
accesses of {\tt x}, and between the ill-scoped atomic accesses of
{\tt y}.
\end{Example}

A resolution proposed by Gaster et al. is to allow the annotated
scopes to differ, as long as both are sufficiently
wide~\cite[\S3.1]{gaster+15}. This enables, for instance, a
$\sdv$-scoped write to synchronise with a $\swg$-scoped read in the
same work-group. The proposal can be formalised in our framework by
changing the definition of the $incl$ relation
(Def.~\ref{def:opencl_derived}) as follows:
\begin{eqnarray*}
incl\textit{1} &\eqdef& ([\swg] ; wg)
\cup ([\sdv] ; dv) \cup ([\sall] ; E^2) 
\\
new\mhyphen{}incl &\eqdef& incl\textit{1} \cap incl\textit{1}^{-1}
\end{eqnarray*}
The idea here is to define a one-sided version of scope inclusion
first, so that $(e_1,e_2)$ is in $incl\textit{1}$ if $e_1$ has a wide
enough scope to `reach' $e_2$. Requiring this to hold in both
directions ensures that both events have sufficient scopes, if not
necessarily the same.


\paragraph{The SC axioms are too weak}
As encoded in our \axOSsimp{} axiom
(Def.~\ref{def:opencl_consistency_axioms}), SC operations in OpenCL
are only guaranteed to provide SC behaviour when one of the $\sccondI$
and $\sccondII$ conditions holds. Since these are conditions on the
whole program, we have a ``clear composability
problem''~\cite{gaster+15}.

We find several reasons why these conditions are problematic. First,
they mean that the default memory scope (which is $\sdv$) is not
sufficient to ensure SC semantics in all situations. Second, any
program that includes a $\swg$-scoped SC atomic, such as
"store(x,1,SC,WG)", immediately violates the conditions. Third, the
conditions are mutually exclusive, so a program that satisfies $\sccondI$ can be combined with another that
satisfies $\sccondII$, with the result satisfying neither. Finally,
consider the following example.

\begin{Example}
\label{ex:opencl_sb} 
The following program, comprising two threads in different work-groups
on the same device, has SC semantics, which means that it cannot
exhibit the relaxed behaviour $"r0" = "r1" = 0$:
\begin{center}
\begin{tabular}{l@{~~}|||@{~~}l}
\multicolumn{2}{c}{\texttt{global atomic\_int *x, *y;}} \\
\evtlbl{$a$}\texttt{store(x,1);} & \evtlbl{$c$}\texttt{store(y,1);} \\
\evtlbl{$b$}\texttt{r0 = load(y);} & \evtlbl{$d$}\texttt{r1 = load(x);}
\end{tabular}
\end{center}
Note that the atomic "store" and "load" operations default to the "SC"
memory order and the $\sdv$ memory scope, and that condition
$\sccondII$ holds. However, if "global" is changed to "global\_fgb",
the relaxed behaviour becomes permissible, because neither condition
$\sccondI$ nor $\sccondII$ holds. Condition $\sccondII$ no longer
holds now that "x" and "y" are in fine-grained atomic SVM buffers, and
condition $\sccondI$ does not hold either because the $\sall$ scope is
not being used. 
\end{Example}

It is jarring that such a small change, from "global" to "global\_"
"fgb", can legitimise relaxed behaviours. Worse still, such a change
may be invisible to the programmer, if they can see only the kernel
code: the assignment of locations to SVM buffers occurs only on the
host side, and such locations are only marked in a kernel as "global".

\paragraph{The SC axioms are too strong}
Following discussion with members of the Khronos OpenCL working group,
we understand that the purpose of condition $\sccondII$ is to enable
efficient implementations of $\sdv$-scoped SC atomics. The intention
of the condition is that if no SC atomic accesses memory shared
between devices, they can be implemented without expensive
inter-device synchronisation. It was thought not to matter that the
specification requires implementations to establish a total order
between SC events on different devices, because it is not possible to
observe this order without creating an inter-device data race.

In fact, this is not the case. We present in
Example~\ref{ex:twisted_sb} a program that satisfies condition
$\sccondII$, and yet is still able to observe the order between SC
events in different devices -- even though these events are
$\sdv$-scoped and access no memory shared between devices.

\begin{Example} 
\label{ex:twisted_sb}
Consider the following program, which comprises two devices, both
executing two threads (stacked vertically). It can be thought of as a
`twisted' version of the store-buffering test.
\begin{center}
\begin{tabular}{l@{~~}||||@{~~}l}
\multicolumn{2}{c}{\texttt{global atomic\_int *x, *y;}} \\
\multicolumn{2}{c}{\texttt{global\_fgb atomic\_int *z1, *z2;}} \\
{\begin{tabular}{@{}l@{}}
\texttt{store(x,1,$\mosc$,$\sdv$);} \\
\texttt{store(z1,1,$\morel$,$\sall$);} \\ \hline\hline
\texttt{r1 = load(z2,$\moacq$,$\sall$)?} \\
\texttt{~~load(x,$\mosc$,$\sdv$) : 1;}
\end{tabular}} & 
{\begin{tabular}{@{}l@{}}
\texttt{store(y,1,$\mosc$,$\sdv$);} \\
\texttt{store(z2,1,$\morel$,$\sall$);} \\ \hline\hline
\texttt{r2 = load(z1,$\moacq$,$\sall$)?} \\
\texttt{~~load(y,$\mosc$,$\sdv$) : 1;}
\end{tabular}}
\end{tabular}
\end{center}
Two threads in different devices write, using $\sdv$ scope, to
distinct "global" locations "x" and "y", and then write to
"global\_fgb" flags, using $\sall$ scope, to signal that they are
done. Meanwhile, two partner threads try to acquire these signals from
the opposite device, and if they are successful, they read the
location their partner (in the same device as they) wrote to. We are
interested in whether these reads can both obtain $0$; that is,
whether the final state $\{"r1"="r2"=0\}$ is allowed.  This final state
could only be obtained via the following execution:
\begin{center}
\begin{tikzpicture}[baseline=10mm,inner sep=1pt]
\node[event, anchor=west](a) at (0,0.75)
{\evtlbl{$a$}$\evW("x",1,\mosc,\sdv)$};

\node[event, anchor=west](c) at (0,0) 
{\evtlbl{$b$}$\evW("z1",1,\morel,\sall)$};

\node[event, anchor=west](d) at (0,-0.75) 
{\evtlbl{$c$}$\evR("z2",1,\moacq,\sall)$};

\node[event, anchor=west](e) at (0,-1.5) 
{\evtlbl{$d$}$\evR("x",0,\mosc,\sdv)$};

\node[event, anchor=west](f) at (4,0.75)
{\evtlbl{$e$}$\evW("y",1,\mosc,\sdv)$};

\node[event, anchor=west](h) at (4,0) 
{\evtlbl{$f$}$\evW("z2",1,\morel,\sall)$};

\node[event, anchor=west](i) at (4,-0.75) 
{\evtlbl{$g$}$\evR("z1",1,\moacq,\sall)$};

\node[event, anchor=west](j) at (4,-1.5) 
{\evtlbl{$h$}$\evR("y",0,\mosc,\sdv)$};

\draw[edgesb] (a) to[auto,swap,pos=0.4] 
node{$sb$} (c.north -| a.south);

\draw[edgesb] (d.south -| a.south) to[auto,swap]
node{$sb$} (e.north -| a.south);

\draw[edgesb] (f) to[auto,swap,pos=0.4] 
node{$sb$} (h.north -| f.south);

\draw[edgesb] (i.south -| f.south) to[auto,swap] 
node{$sb$} (j.north -| f.south);

\draw[edgerf] (c.east) to[auto,swap,pos=0.2]
node{$rf$} (i.west);

\draw[edgerf] (h.west) to[auto, pos=0.2]
node{$rf$} (d.east);

\draw[edgerb] (e.west) to[auto, bend left=30]
node{$fr$} (a.west);

\draw[edgerb] (j.east) to[auto, swap, bend right=70]
node{$fr$} (f.east);

\begin{pgfonlayer}{wglayer}
\node[wg, fit=(a)(c)] (wg1) {};
\node[wg, fit=(d)(e)] (wg2) {};
\node[wg, fit=(f)(h)] (wg4) {};
\node[wg, fit=(i)(j)] (wg5) {};
\end{pgfonlayer}

\begin{pgfonlayer}{dvlayer}
\node[dv, fit=(wg1)(wg2)] {};
\node[dv, fit=(wg4)(wg5)] {};
\end{pgfonlayer}

\end{tikzpicture}
\end{center}
where the outer dotted rectangles delimit $dv$ equivalence classes and
the inner ones delimit $thd$ equivalence classes.

The execution is inconsistent, and therefore must be forbidden by a
compiler. To see this, observe that each $rf$ edge induces a
synchronisation ($gsw$) edge, and hence global happens-before. Since
$sb$ edges also contribute to global happens-before, we obtain the
cycle
$a\xrightarrow{ghb} b\xrightarrow{ghb} g\xrightarrow{ghb}
h\xrightarrow{fr} e\xrightarrow{ghb} f\xrightarrow{ghb}
c\xrightarrow{ghb} d\xrightarrow{fr}a$.
This makes the execution fall foul of \axOSsimp, which is non-vacuous
here because the condition $\sccondII$ is satisfied.
\end{Example}

That the execution in the example above is not allowed implies that
OpenCL implementations must make the order of SC write operations
visible to all devices, even when those writes are only performed with
$\sdv$ scope. In other words, the current phrasing of the OpenCL
memory model demands too much from the compiler-writer to permit an
efficient implementation of $\sdv$-scoped SC atomics, while in other
respects offering too little to the programmer, by guaranteeing SC
semantics only when an onerous condition holds.

To summarise: the intent of the Khronos working group was to enable
efficient implementation of $\sdv$-scoped SC atomics by compilers, at
the expense of programmer inconvenience.  Instead, our formalisation
shows that we have the worst of both worlds: the programmer is
inconvenienced, and yet a correct compiler is obliged to enforce
inter-device orderings on $\sdv$-scoped SC atomics.

\section{Overhauling the SC Axioms in OpenCL}
\label{sec:opencl_sc}

We describe how the handling of SC atomics in
OpenCL can be changed to address the shortcomings identified in \S\ref{sec:opencl_problems}.

Building on a suggestion by Gaster et
al.~\cite[\S7.2]{gaster+15}, we propose to eradicate the stringent
conditions on the existence of the SC order by simply intersecting the
constraints on the SC order with the scope-inclusion relation. This
essentially means that the orderings imposed between events by the SC
axioms only take effect if those events have inclusive scopes. Under
this proposal, which recalls the way C11's synchronisation relation
($sw$, Def.~\ref{def:c11_further_derived}) is intersected with
scope-inclusion when producing OpenCL's version ($rsw$,
Def.~\ref{def:opencl_derived}), we do not need to restrict the
programmer's usage of SC atomics to certain scopes; instead, the
guarantees provided by those SC atomics degrade gracefully as their
scopes narrow.
\begin{definition}[Proposed SC axiom for OpenCL]
The following axiom for SC atomics in OpenCL is obtained from \axOSsimp{} by removing
the $\sccondI$ and $\sccondII$ conditions and instead
intersecting with $incl$:
\begin{gather}
\tag{\axOSscoped}
\acyclic(\mosc^2 \cap (Fsb^?; (ghb \cup lhb \cup fr \cup mo) ; sbF^?) \cap \mhl{incl})
\end{gather}
\end{definition}

\subsection{Effect on the Standard} 
\label{sec:opencl_proposal}
To accommodate our proposal, we propose that the wording of the OpenCL
2.1 standard \cite[51/14--31 and 51/34--52/13]{opencl21} be changed to
match the text given in \S\ref{sec:c11_proposal}, but with `happens
before' replaced with `global or local happens before', and
`consistent with the SC-before order' replaced with `consistent with
the SC-before order restricted to operations with inclusive
scopes'. This replaces \OoldWORDS{} words with \OnewWORDS{} words,
while retaining the standard's style and terminology.

\subsection{Implementability of the New SC Axiom}
\label{sec:opencl_implementability}

The new \axOSscoped{} axiom is stronger than the original
\axOSsimp{} axiom, so we must confirm that our
proposal does not place undue demands on compilers that implement the
memory model. 

The only published compilation scheme of the OpenCL 2.0 memory model
of which we are aware is that published by AMD~\cite{orr+15} and later
formalised by Wickerson et al.~\cite{wickerson+15a}. The scheme
compiles the release/acquire fragment of OpenCL atomics, and its
soundness has been verified against an operational model of an AMD
GPU~\cite{wickerson+15a}. In this subsection we describe how the
scheme can be extended to support SC atomics, and we demonstrate via a
series of examples that the extended scheme meets the requirements of
our revised SC axiom. The original compilation scheme does not cater
for multiple devices, and does not include fences, and we do not
attempt here to extend the scheme to cover these features. As such,
this scheme does not engage directly with the problems of inter-device
SC atomics that we noted in the previous section; however, it does
illustrate how $\swg$- and $\sdv$-scoped SC atomics can co-exist.

\begin{table}[t]
\centering
\begin{tabular}{l@{~~~}ll}
& OpenCL atomic operation & Assembly instructions
\\ \hline
\ding{182} & \tstack{$r = \texttt{load}(x,\mosc,\swg)$} & 
\tstack{$\INSld\,r\,x$} 
\\ 
\ding{183} & \tstack{$r = \texttt{load}(x,\mosc, \sdv)$} & 
\tstack{
\hl{$\INSinvall1$} {\tt ;}  
$\INSld\,r\,x$ {\tt ;}
$\INSinvall1$} 
\\ 
\ding{184} & \tstack{$\texttt{store}(x, r, \mosc, \swg)$} & 
\tstack{$\INSst\,r\,x$} 
\\ 
\ding{185} & \tstack{$\texttt{store}(x, r, \mosc, \sdv)$} & 
\tstack{
$\INSflushl1$ {\tt ;}
$\INSst\,r\,x$ {\tt ;}
\hl{$\INSflushl1$}} 
\\ 
\ding{186} & \tstack{$r = \texttt{fetch\_inc}(x, \mosc, \swg)$} & 
\tstack{$\INSincl1\,r\,x$} 
\\ 
\ding{187} & \tstack{$r = \texttt{fetch\_inc}(x, \mosc, \sdv)$} & 
\tstack{$\INSflushl1$ {\tt ;} 
$\INSincl2\,r\,x$ {\tt ;} 
$\INSinvall1$} 
\\ 
\end{tabular}
\caption{Compiling the revised OpenCL memory model}
\label{tab:compilation_scheme}
\end{table}

\paragraph{The AMD compilation scheme}
The operational model is quite simple. Each work-group has its own L1
cache, and each device has its own L2 cache. Since the compilation
scheme considers only the single-device case, the L2 cache can be
safely thought of as the main memory. No instruction reordering is
permitted. At any time, the environment can flush a dirty L1 cache
entry to the L2 (and thereby make it clean), can fetch an L2 entry to
replace a clean L1 entry, and can evict a clean L1 entry.

The semantics of the various assembly instructions can be summarised
as follows. $\INSld\,r\,x$ loads into register $r$ from the nearest
cache that contains a valid entry for $x$; $\INSst\,r\,x$ stores from
$r$ into the local L1 cache, first flushing $x$'s entry therein if it
is invalid; $\INSincl1\,r\,x$ increments $x$ in the local L1 cache;
$\INSincl2\,r\,x$ increments $x$ in the L2 cache, first flushing any
dirty entry for $x$ in the local L1 cache; $\INSflushl1$ flushes all
dirty entries in the local L1 cache; and $\INSinvall1$ marks all
entries in the local L1 cache as invalid.

The extensions to the compilation scheme are given in
Tab.~\ref{tab:compilation_scheme}. Here, "fetch\_inc" stands for
`atomic fetch and increment', and provides a representative of RMW
operations in OpenCL.

\paragraph{Correctness of the compilation scheme}
Most of the flush and invalidate instructions in the compilation
scheme are necessary to ensure correct release/acquire semantics. For
SC atomics, we need add only two further instructions: the
$\INSinvall1$ before the load in row \ding{183}, and the $\INSflushl1$
after the store in row \ding{184}. The need for these instructions can
be motivated by considering the following two examples, which
correspond to the classic store-buffering and IRIW litmus tests.

The memory model requires the program in Example~\ref{ex:opencl_sb}
not to produce the final state $"r0"="r1"=0$. With only
release/acquire semantics, the compilation scheme inserts no flush or
invalidate instructions between the "store" and the "load" in each
thread, and the relaxed behaviour can be observed: both threads might
pre-fetch $"x"="y"=0$ into their respective L1 caches (the threads are
in different work-groups, so they have different L1 caches), then
perform their stores, and finally load the L1-cached values of "x" and
"y". However, placing a $\INSflushl1$ after the store and an
$\INSinvall1$ before the load ensures that no sequence of fetching and
flushing can lead to the relaxed behaviour. We do not need
$\INSflushl1$ or $\INSinvall1$ instructions before or after the SC
increment instruction, because $\INSincl2$ writes directly to the L2,
invalidating the L1 as it does so.

The memory model also requires the IRIW litmus test
\begin{center}
\begin{tabular}{@{}l@{~}|||@{~}l@{~}|||@{~}l@{~}|||@{~}l@{}}
\multicolumn{4}{c}{\texttt{global atomic\_int *x; global atomic\_int *y;}} \\
\texttt{store(x,1);} &
\texttt{store(y,1);} & 
\texttt{r0=load(x);} &
\texttt{r2=load(y);}
\\
& 
& 
\texttt{r1=load(y);} &
\texttt{r3=load(x);}
\end{tabular}
\end{center}
not to produce the final state $\{"r0"="r2"=1, "r1"="r3"=0\}$. (Recall
that these "store" and "load" operations use memory order $\mosc$ and
scope $\sdv$ by default.) Here, an $\INSinvall1$
instruction between each pair of loads is sufficient to rule out such
executions.\footnote{On weaker models, such as Power, that are not
\emph{multi-copy atomic}~\cite{stone+95}, further synchronisation
would be required between the loads.}

\section{Simulating the Memory Models with \herd{}}
\label{sec:herd}

Our overhaul of SC atomics avoids the requirement for the $S$ relation
to be explicitly constructed in execution witnesses. Our hypothesis
was that this would lead to improved efficiency in the process of
exhaustively enumerating the allowed behaviour of litmus tests that
use SC atomics. We now explain how we extended the \herd{} memory
model simulator in order to enable investigation of C11 and OpenCL
litmus tests (\S\ref{sec:extensionstoherd}), and present experimental
results using \herd{} to compare the efficiency of simulation before
and after our overhaul, and also in comparison to the \CDSChecker{}
memory model simulator~\cite{norris+13} (\S\ref{sec:experiments}). For
a family of litmus tests derived from Dekker's algorithm, our results
show that our revised axioms lead to an exponential speedup in
simulation time using \herd{}, bringing performance using \herd{},
which is general-purpose and exhaustive on loop-free programs, much
closer to that of \CDSChecker{}, which is specifically tuned for the
C11 memory model and is not guaranteed to be exhaustive, even on
loop-free programs.

\subsection{Extensions to \herd{}}\label{sec:extensionstoherd}

The version of \herd{} described by Alglave et al.~\cite{alglave+14, alglave+15}
supports only assembly code: sequences of labelled instructions and
gotos. In order to simulate our formalisations of the C11 and OpenCL
memory models, we have extended the \cat{} format to support the
definition of $\faulty$ axioms, and the \herd{} tool with both a
routine for alerting the user when a faulty execution is detected and
a module for translating C11 and OpenCL programs into their executions.

We model only a small fragment of the C11 language: enough to encode
the litmus tests we found useful for testing our formalisation. We
exclude, for example, the address-of operator, compound types, and function calls.
We include \texttt{if} and \texttt{while} blocks, pointer
dereferencing, simple expressions, and built-in atomic functions such
as \texttt{atomic\_thread\_fence} (C11) and
\texttt{atomic\_work\_item\_fence} (OpenCL). 

\newcommand\thisHASbeenCUT{
A notable aspect of our extension to \herd{} is our faithful modelling
of compare-and-swap operations. We model C11's
\[ 
\stack{\texttt{atomic\_compare\_exchange\_}\\\qquad\texttt{strong\_explicit}(\var{obj},\var{exp},\var{des},\var{succ},\var{fail})}
\]
instruction (where $\var{obj}$ is an atomic location, $\var{exp}$
points to the value that $\var{obj}$ is expected to contain,
$\var{des}$ is the desired new value, $\var{succ}$ is the memory order
to use for the RMW operation in the case of success, and $\var{fail}$
is the memory order to use when loading from $\var{obj}$ in the case
of failure) as having the following executions:
\[
\stack{\left\{
\begin{array}{@{}l|l@{}}
\begin{tikzpicture}[inner sep=1pt, baseline=6.2mm]
\node[event](n1) at (0,1.4) {$\evRna(\var{exp},v_{\rm exp})$};
\node[event](n2) at (0,0.7) {$\evR(\var{obj},v_{\rm obj},\var{fail})$};
\node[event](n3) at (0,0) {$\evWna(\var{exp},v_{\rm obj})$};
\draw[edgesb] (n1) to[auto,pos=0.6] node{$sb$} (n2);
\draw[edgesb] (n2) to[auto,pos=0.6] node{$sb$} (n3);
\end{tikzpicture}
& \begin{array}{@{}l@{}} v_{\rm exp}~\text{is a value} \\ v_{\rm
obj}~\text{is a value} \\ v_{\rm exp} \neq v_{\rm obj} \end{array}
\end{array}
\right\} \cup {} \\ ~~~~~~~~~~~~~~~~~~~~~~~~
\left\{
\begin{array}{@{}l|l@{}}
\begin{tikzpicture}[inner sep=1pt, baseline=2.7mm]
\node[event](n1) at (0,0.7) {$\evRna(\var{exp},v_{\rm exp})$};
\node[event](n2) at (0,0) {$\evRMW(\var{obj},v_{\rm exp},\var{des},\var{succ})$};
\draw[edgesb] (n1) to[auto,pos=0.6] node{$sb$} (n2);
\end{tikzpicture}
& v_{\rm exp}~\text{is a value}
\end{array}
\right\}}
\]
and we omit the `$v_{\rm exp} \neq v_{\rm obj}$' condition to model
the {\tt weak} version, which may fail
spuriously.\footnote{\citec{7.17.7.4}{}} These instructions are rather
subtle to model because the memory order that must be used to access
the $\var{obj}$ location depends on the value $\var{obj}$ contains,
which of course can only be determined having performed the load. To
our knowledge, previous work on C11 memory model simulation has either
modelled these instructions incorrectly~\cite{batty+11, norris+13} or
not at all~\cite{alglave+14, vafeiadis+15}.
}

\subsection{Simulating the C11 Model: Performance Evaluation}
\label{sec:experiments}

\begin{figure}
\centering

\pgfplotstableset{
    create on use/mean/.style={create col/mean},
    create on use/stddev/.style={create col/standard deviation},
    create on use/stderror/.style={create col/standard error}
}

\begin{tikzpicture}
\begin{semilogyaxis}[
  xlabel={Number of threads, $N$}, 
  ylabel=Simulation time /s, 
  legend pos=south east,
  ymin=0.005,
  ymax=3000,
  height=60mm,
  width=85mm,
]
\addplot+[color=green, mark=square] 
table [x=N, y=mean] 
{herd_tests/c11_orig.cat-results.txt};

\addplot+[color=blue, mark=x] 
table [x=N, y=mean] 
{herd_tests/c11_simp.cat-results.txt};

\addplot+[color=red, mark=+] 
table [x=N, y=mean] 
{CDSChecker_tests/CDSChecker-results.txt};

\legend{{\Herd, original}, {\Herd, proposed}, {\CDSChecker}}
\end{semilogyaxis}
\end{tikzpicture}
\caption{Time to simulate an $N$-threaded store-buffering test}
\label{fig:graph_time}
\end{figure}
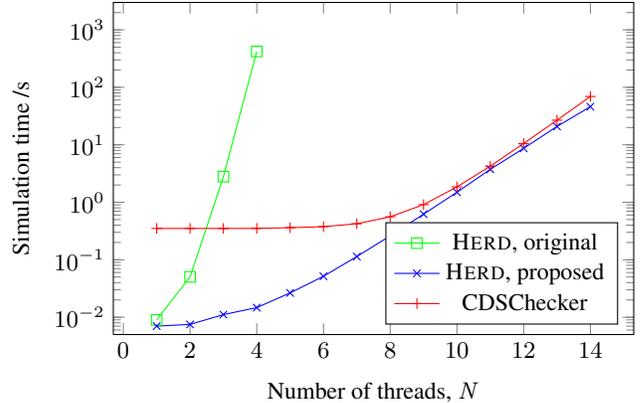

We now compare the performance of \herd{} in enumerating the behaviours of litmus
tests (a) when equipped with the original SC axioms in C11 vs.\ (b) when
equipped with our revised SC axioms.  We also provide performance results gathered using
\CDSChecker, a custom-built simulator for the C11 memory model~\cite{norris+13}.
The \herd{} tool guarantees exhaustive enumeration of allowed
behaviours for a loop-free litmus test;
\CDSChecker{} aims for high coverage of behaviours, but is known to be non-exhaustive in general~\cite{norris+13}.

Recall that Dekker's mutual exclusion algorithm~\cite{dijkstra02}
is a key use case for SC atomics.
The essential idiom underlying an $N$-threaded version of Dekker's algorithm
is captured by the following $N$-threaded store-buffering
litmus tests:
\begin{center}
$P_N~\eqdef~\Big(\text{\footnotesize\begin{tabular}{@{}l@{~}||@{~}l@{~}||@{~}l@{~}||@{~}l@{}}
\texttt{store(x$_2$,1);} &
\texttt{store(x$_3$,1);} & 
\dots &
\texttt{store(x$_1$,1);} 
\\
\texttt{r$_1$=load(x$_1$);} & 
\texttt{r$_2$=load(x$_2$);} & 
\dots & 
\texttt{r$_{N}$=load(x$_{N}$);}
\end{tabular}}\Big)$
\end{center}
that operate on a collection $\{"x"_1,\dots,"x"_N\}$ of atomic
locations initialised to zero. Recall that atomic "store" and "load"
operations use memory order $\mosc$ by default.  Dekker's algorithm
requires that it is \emph{not} possible to observe the final state
where $"r"_1=\dots="r"_N=0$; only $\mosc$ is strong enough to
rule out this relaxed behaviour.

We use the family $(P_N)_{N \in \mathbb{N}}$ to assess the scalability
of the two versions of \herd{} and of \CDSChecker.
Figure~\ref{fig:graph_time} shows the time each tool takes to simulate
$P_N$ as $N$ increases.\footnote{We used \herd{} revision 88ff189
(\url{http://github.com/herd/} \url{herdtools}) and 
\CDSChecker{} revision 7c51087 (\url{git://demsky.} \url{eecs.uci.edu/model-checker.git}).}
Experiments were conducted on a 3.1\ GHz MacBook Pro, and each data
point represents the mean of ten runs. We do not include error bars
because the standard deviation is negligible. The original memory
model, naively implemented in \herd{}, times out on just $4$ threads.
This is because it iterates over all $(2N)!$ orders of the $2N$ SC
events that are in every execution of $P_N$. When \herd{} is provided
with our revised memory model, simulation times greatly improve.
Bearing in mind the logarithmic y-axis, the performance of both
\herd{} on the revised memory model and \CDSChecker{} appears to scale
exponentially with $N$, which meets expectations since $P_N$ has
$2^N - 1$ unique final states. Still, \CDSChecker{} significantly
outperforms \herd{} when simulating $P_N$, and on several other
programs that we tried. This is because \CDSChecker{}, unlike \herd{},
is optimised specifically for the C11 memory model, through the use of
such techniques as the early elimination of infeasible executions, and
a variant of dynamic partial order reduction
(DPOR)~\cite{flanagan+05a} on the $S$ order. In fact, we conjecture
that the use of DPOR here has an effect similar to our proposal to
rephrase the memory model with $S$ as a partial order.

Figure~\ref{fig:graph_time} demonstrates that simply by tweaking the
axioms that define the memory model, simulation time can be
dramatically decreased, without the need to implement complex
optimisations, such as DPOR, that make it difficult to assess the
soundness and completeness of the tool. It happens that \CDSChecker
\emph{is} exhaustive on all of our $P_N$ programs, but we remark that
we can only be sure of this because of \Herd{}.




\section{Related Work}
\label{sec:related}

\paragraph*{The C11 memory model} has been formalised several times.
Batty et al.~\cite{batty+11} present a comprehensive formalisation
using Lem~\cite{mulligan+14}. Vafeiadis et
al.~\cite{vafeiadis+13,vafeiadis+15} and Batty et al.~\cite{batty+13}
have also formalised slightly simplified variations. Alglave et al.
have formalised a release/acquire fragment of the C11 model (without
release sequences, fences, non-atomics, or data races) in the \cat{}
language, and have shown it to be an instance of their generic
axiomatic memory model~\cite{alglave+14}. We use the \cat{} language
in our work too, but our comprehensive model, which incorporates
undefined behaviours and a richer language of events, no longer fits
within their generic framework.

We remark that in the absence of fences, our \axSsimp{} axiom
(see Theorem~\ref{thm:SC_simp}) forbids
the same dependency cycles that Shasha and Snir characterise as
violations of sequential consistency~\cite{shasha+88}. In a sense, one
contribution of our paper is to simplify the semantics of C11's SC
atomics to the point where it can be defined, for the first time, in
the Shasha--Snir style.

\paragraph{Criticisms of the C11 model}
Batty et al. describe a fundamental problem in the structure of the
C11, C++11, C++14 and OpenCL memory models: the so-called
``thin-air'' executions~\cite{batty+15}. This is a difficult open
problem requiring a radically different approach; we do not address it
here. 

Vafeiadis et al. note that the current rules governing SC
atomics break desirable properties of the memory model, harming the
prospect of reasoning above it, and they propose a strengthening of
the model to fix this~\cite{vafeiadis+15}. Our proposal builds on
theirs (\S\ref{sec:sc-partial}), but goes further
(\S\ref{sec:sc-simp}), arriving at a substantially simpler model. A
similar proposal was in fact considered by Vafeiadis et al. in the
context of the original total-order SC axioms~\cite{vafeiadis+15}, but
abandoned over concerns that it would invalidate the existing Power
compilation scheme. In our work, we have demonstrated that such a
proposal \emph{is} in fact valid on Power (and x86).

We note that despite our strengthening, SC fences remain too weak to
restore sequential consistency in all circumstances, even when placed
between every pair of accesses. This weakness was intentional in C11
to permit efficient implementation over Intel's Itanium
architecture~\cite{batty14}, but it does harm programmability. Lahav
et al.~\cite{lahav+16} have proposed an alternative implementation of
SC fences, in terms of acquire/release RMWs on a distinguished
location, that always restores sequential consistency.

\paragraph*{The OpenCL 2.0 memory model} has recently been described
by Gaster et al.~\cite{gaster+15}, as an instance of a heterogeneous
race-free (HRF) model~\cite{hower+14}. Our work improves on theirs in
several ways. A key shortcoming of their work is its relative
informality: it lacks the mathematical precision that is required to
resolve all the details of the OpenCL memory model. Our formalisation,
in contrast, is precise enough to be executed by a machine
(cf.~\S\ref{sec:herd}). Moreover, their characterisation of the OpenCL
memory model has several technical issues. It replaces the
specification's modification order (which orders atomic write events)
with a \emph{coherence order} (which orders both read \emph{and} write
events) without proving that the intent of the specification is
preserved by this change. Another infidelity to the specification is
the omission of release sequences, which prohibits the correct
treatment of release-fences. Indeed, Gaster et al. include no formal
treatment of fences at all, describing their behaviour only in prose.
Our \cat{} presentation of the OpenCL memory model treats
release sequences and fences in full.
Its informality aside, Gaster et al.'s work contains numerous insights
into the design and workings of the OpenCL memory model, and provided
a valuable basis for our formalisation efforts. 

We have already begun to build on top of the formalisation of the
OpenCL memory model presented here, as part of our investigations into
the semantics of a proposed extension to OpenCL called
\emph{remote-scope promotion}~\cite{wickerson+15a}. That work, which
has already been published, describes only a small `release/acquire'
fragment of the OpenCL memory model, while the current paper describes
the full model, including the interesting and important SC and relaxed
atomics.

\paragraph{Implementations of the OpenCL memory model} AMD and Intel
have recently released OpenCL 2.0-compliant
implementations~\cite{amd-developer-central15,intel-developer-zone14}.
We are aware only of one implementation of the OpenCL memory model
that has been formalised: namely, a compilation scheme from OpenCL
(extended with a feature called remote-scope promotion~\cite{orr+15})
to a model of next-generation AMD GPUs~\cite{wickerson+15a}. Alglave
et al. present an experimentally-validated axiomatic model of an
Nvidia GPU~\cite{alglave+15}, which could provide another compilation
target for our OpenCL memory model. However, we find that their model
is too weak to admit an efficient mapping from OpenCL. Specifically,
it does not provide the property of \emph{cumulativity}:
synchronisation at one scope cannot be chained with further
synchronisation at a wider scope to induce overall synchronisation
between the two end-points. Since cumulativity is a property required
by the OpenCL memory model, we deduce that the OpenCL compiler must,
very expensively, treat all operations as having the widest scope.

\paragraph*{Memory model simulators} other than \herd{} that are
capable of handling the C11 model include \Cppmem~\cite{batty+11},
Nitpick~\cite{blanchette+11} and \CDSChecker{}~\cite{norris+13}. We
did not include \Cppmem{} and Nitpick in our tool comparison
(\S\ref{sec:experiments}) because Norris et al. have already
demonstrated that \CDSChecker's performance is far
superior~\cite{norris+13}.

Because it is highly optimised for the C11 memory model,
\CDSChecker{} continues to outperform \herd{} even on the revised model.
\Herd{} on the other hand is deliberately designed \emph{not} to be
optimised for a particular model, but to be instead a generic memory
model simulator. A key advantage of using a generic memory model
simulator like \herd{} is that it is easy to tinker with the model
during the development process: one must only modify a text file and
restart \herd{} in order to explore the impact of a proposed change.
Indeed, this ease of modification, together with the challenge of
expressing the C11 model in the very concise \cat{} language, inspired
our discovery of the simpler SC axioms described in this paper.
Moreover, where \CDSChecker{} is designed for efficiency, sometimes
at the cost of fidelity to the memory model (the lack of
self-satisfying conditionals, for instance, is a source of
incompleteness in \CDSChecker), our formalisation and simulator
are designed primarily to represent the memory model as closely as
possible.

\CDSChecker{} obtains its main performance benefits by exploring
partial modification orders. It is therefore natural to ask whether the
memory model could be revised to accommodate partial modification
orders in the same way that we have incorporated a partial $S$
order. We believe that this is not straightforwardly possible without
changing the model: our partial order reduction on $S$ hinges on its
constraints all having the form $\irreflexive(S ; r)$ for some $r$,
but this is not the case for $mo$ -- see axiom \axiom{Rmw}
(Def.~\ref{def:consistency_axioms}) for instance.

\section{Conclusion}

Our overhaul of the semantics of SC atomics and fences provides four
main benefits in relation to the C11 and OpenCL memory models: more
efficient exploration of the behaviours of litmus tests
(cf.~\S\ref{sec:sc-partial}, \S\ref{sec:experiments}); refined specification text that we argue
is easier for programmers and compiler-writers to understand
(cf.~\S\ref{sec:c11_proposal}); improved usability of the languages by
programmers (cf.~\S\ref{sec:opencl_sc}); and opportunities for
compiler-writers to produce more efficient implementations
(cf.~\S\ref{sec:opencl_sc}). We argue that our proposed changes to the memory models
validate all of the formalised C11 and OpenCL
compilation schemes of which we are aware.

A topic for future research is the consideration of memory consistency
between OpenCL devices and the host application that launches kernels
on these devices; our treatment in this paper focuses solely on
interactions between kernel threads. We also plan to use our memory
model as a basis for reasoning about OpenCL programs, extending the
capabilities of tools such as GPUVerify~\cite{betts+15}, where
existing support for atomic operations is limited and not based on
formal foundations~\cite{bardsley+14a}.

\section*{Acknowledgements} 
\def\georgeI{EP/I020357/1} \def\georgeII{EP/K015168/1}
\def\georgeIII{EP/I01236/1} \def\allyI{EP/K011499/1} Luc Maranget
kindly advised on our extensions to \herd{}. We thank Jade Alglave,
Nathan Chong, Benedict Gaster, Vinod Grover, Lee Howes, Jeroen Ketema,
Matthew Parkinson, Peter Sewell, Tyler Sorensen, and our anonymous
reviewers for their feedback and encouragement. This work was supported
by the EPSRC (grants \allyI, \georgeI, \georgeII, and \georgeIII), and
by the EU FP7 project CARP (project number 287767).

\appendix

\section{Proof of Theorem~\ref{lem:soundnesslemma}}
\label{appx:proofs}

The following theorem states that the x86 and Power compilation schemes
for C11, as given in Tab.~\ref{tab:x86PowerComp}, remain sound in the
presence of our revised SC axiom, \axSsimp.

\begin{table}[h]
\centering
\begin{tabular}{l@{~~}l@{~~}l@{~~}l}
& C11 operation & x86  & Power
\\ \hline
\ding{182} & \tstack{$r = \texttt{load}(x,\mosc)$} & 
\textsf{lock xadd}(0) & \textsf{sync; ld; cmp; bc; isync}
\\ 
\ding{183} & \tstack{$\texttt{store}(x, r, \mosc)$} & 
\textsf{lock xchg} & \textsf{sync; st}
\\ 
\ding{184} & \tstack{$r = \texttt{fence}(x, \mosc)$} & 
\textsf{mfence} & \textsf{sync}
\\ 
\end{tabular}
\caption{Compiling the C11 SC atomics}
\label{tab:x86PowerComp}
\end{table}

{
\def\thetheorem{\ref{lem:soundnesslemma}}
\begin{theorem}[repeated from \S\ref{sec:c11_soundness}]
\soundnesslemma
\end{theorem}
}

\begin{proof}[Proof (x86 case)]
The axiomatic model of Owens et al. restricts the behaviour of memory
using a partial order over x86 memory events called $\memord$. The
proof of Batty et al.~\cite{batty+11} constructs the relations of the
C11 execution using $\memord$; modification order ($mo$) and
reads-from ($rf$), in particular, are projected from it. Here we rely
on several properties of $\memord$ as set out by Owens et al.: when
restricted to writes, it is a linear order; program order between two
events is included in memory order if there is an intervening fence or
if either instruction is locked; program-order edges from reads to
later events are included; and a read observes the most recent
preceding write in $\memord$.

We proceed by contradiction, showing that given the construction of
$rf$ and $mo$ used in the proof of Batty et al., any cycle in
$\mosc^2 \setminus id \cap (Fsb^?; (hb \cup fr \cup mo) ; sbF^?)$
implies either the existence of a cycle in $\memord$, or an
inconsistent $rf$ edge.

Any cycle in the relation is made up of $mo$, $hb$ and $fr$ edges,
possibly linked with $sb$ edges.  The $mo$, $hb$ and $sb$ edges all
imply corresponding $\memord$ edges.  To see this, note:
\begin{itemize}
\item $\memord$ is a linear order over writes;
\item any $hb$ edge in the \axSsimp{} relation
begins with either a fence or a locked instruction, so $sb$ edges
correspond to $\memord$ edges, then there may be a chain of edges in
$mo;rf;sb$, where the final $sb$ edge is headed by a read, so
transitivity implies that $hb$ corresponds to $\memord$; and
\item  any $sb$ edges are either between locked instructions
  or have a fence between accesses, and so correspond to $\memord$
  edges.
\end{itemize}
Finally, if the cycle contains an $fr$ edge, then $\memord$ cannot
contradict this: the read would become inconsistent in the x86
execution. Then for a given \axSsimp{} cycle, we have a sequence of
$\memord$ edges that would form a cycle if not for holes corresponding
to $fr$ edges.

We now show that there is indeed a cycle in $\memord$. Consider an
$fr$ edge: its head is a read, so the preceding edge must either be a
$sb$ edge or a $hb$ edge. If it is a $sb$ edge then either the head of
that edge is a write, or the edge that precedes that is an $hb$
edge. In all cases, the read at the head of the $fr$ edge is preceded
in $hb$ by a write.  As $\memord$ is total over writes, it must order
this preceding write's x86 counterpart before the write in the tail of
the $fr$ edge. We use this fact to construct a cycle in $\memord$, a
contradiction.
\end{proof}

\begin{proof}[Proof (Power case)]
In their proof, Batty et al. construct a C11 execution from a Power
trace such that Power coherence and $rf$ edges match their constructed
C11 counterparts. In proving the SC axioms hold over the execution,
they prove a property, $\goodsc$, that establishes a total order over
the SC atomics of the execution that contains $po$, $co$, $fr$ and an
extended variant of reads from, $erf$, each restricted to the SC
atomics. The SC fences are added to this relation in a way that is
consistent with $co$ and $rf$ in the rest of the execution, preserving
the invariant that it is a strict partial order. The following edges
become part of the SC order:
$[\mosc];po^?;(rf^{-1})^?;co;rf^?;po^?;[F \cap \mosc]$ and
$[F \cap \mosc];po^?;(rf^{-1})^?;co;rf^?;po^?;[\mosc]$.
The proof goes on to show that $hb$ restricted to the SC actions is a
subset of the total order.  The construction of $mo$ and $rf$ in the
C11 execution follow the Power trace directly, so $\goodsc$ together
with the addition of the fence edges, which contain $Fsb^?;fr; sbF^?$
and $Fsb^?;co; sbF^?$, show the acyclicity of
$\mosc^2 \setminus id \cap (Fsb^?; (hb \cup fr \cup mo) ; sbF^?)$
directly.
\end{proof}

\bibliographystyle{abbrvnat}

\bibliography{popl}

\begin{thebibliography}{42}
\providecommand{\natexlab}[1]{#1}
\providecommand{\url}[1]{\texttt{#1}}
\expandafter\ifx\csname urlstyle\endcsname\relax
  \providecommand{\doi}[1]{doi: #1}\else
  \providecommand{\doi}{doi: \begingroup \urlstyle{rm}\Url}\fi

\bibitem[Alglave et~al.(2010)Alglave, Maranget, Sarkar, and Sewell]{alglave+10}
J.~Alglave, L.~Maranget, S.~Sarkar, and P.~Sewell.
\newblock Fences in weak memory models.
\newblock In \emph{CAV}, 2010.

\bibitem[Alglave et~al.(2014)Alglave, Maranget, and Tautschnig]{alglave+14}
J.~Alglave, L.~Maranget, and M.~Tautschnig.
\newblock Herding cats: modelling, simulation, testing, and data-mining for
  weak memory.
\newblock \emph{TOPLAS}, 2014.

\bibitem[Alglave et~al.(2015)Alglave, Batty, Donaldson, Gopalakrishnan, Ketema,
  Poetzl, Sorensen, and Wickerson]{alglave+15}
J.~Alglave, M.~Batty, A.~F. Donaldson, G.~Gopalakrishnan, J.~Ketema, D.~Poetzl,
  T.~Sorensen, and J.~Wickerson.
\newblock {GPU} concurrency: weak behaviours and programming assumptions.
\newblock In \emph{ASPLOS}, 2015.

\bibitem[{AMD Developer Central}(2015)]{amd-developer-central15}
{AMD Developer Central}.
\newblock \emph{{AMD APP SDK 3.0 released, featuring OpenCL 2.0}}, 2015.
\newblock URL \url{http://developer.amd.com/community/blog/2015/08/26/
  introducing-app-sdk-30-opencl-2/}.

\bibitem[Bardsley and Donaldson(2014)]{bardsley+14a}
E.~Bardsley and A.~F. Donaldson.
\newblock Warps and atomics: {B}eyond barrier synchronization in the
  verification of {GPU} kernels.
\newblock In \emph{NASA Formal Methods}, 2014.

\bibitem[Batty(2014)]{batty14}
M.~Batty.
\newblock \emph{The C11 and C++11 Concurrency Model}.
\newblock PhD thesis, University of Cambridge, October 2014.

\bibitem[Batty et~al.(2011)Batty, Owens, Sarkar, Sewell, and Weber]{batty+11}
M.~Batty, S.~Owens, S.~Sarkar, P.~Sewell, and T.~Weber.
\newblock Mathematizing {C++} concurrency.
\newblock In \emph{POPL}, 2011.

\bibitem[Batty et~al.(2012)Batty, Memarian, Owens, Sarkar, and
  Sewell]{batty+12}
M.~Batty, K.~Memarian, S.~Owens, S.~Sarkar, and P.~Sewell.
\newblock Clarifying and compiling {C/C++} concurrency: from {C++11} to
  {POWER}.
\newblock In \emph{POPL}, 2012.

\bibitem[Batty et~al.(2013)Batty, Dodds, and Gotsman]{batty+13}
M.~Batty, M.~Dodds, and A.~Gotsman.
\newblock Library abstraction for {C/C++} concurrency.
\newblock In \emph{POPL}, 2013.

\bibitem[Batty et~al.(2015)Batty, Memarian, Nienhuis, Pichon-Pharabod, and
  Sewell]{batty+15}
M.~Batty, K.~Memarian, K.~Nienhuis, J.~Pichon-Pharabod, and P.~Sewell.
\newblock The problem of programming language concurrency semantics.
\newblock In \emph{ESOP}, 2015.

\bibitem[Batty et~al.(2016)Batty, Donaldson, and
  Wickerson]{overhauling-companion}
M.~Batty, A.~F. Donaldson, and J.~Wickerson.
\newblock Overhauling {SC} atomics in {C11} and {OpenCL} -- companion webpage,
  2016.
\newblock URL \url{http://multicore.doc.ic.ac.uk/overhauling}.

\bibitem[Betts et~al.(2015)Betts, Chong, Donaldson, Ketema, Qadeer, Thomson,
  and Wickerson]{betts+15}
A.~Betts, N.~Chong, A.~F. Donaldson, J.~Ketema, S.~Qadeer, P.~Thomson, and
  J.~Wickerson.
\newblock The design and implementation of a verification technique for {GPU}
  kernels.
\newblock \emph{TOPLAS}, 2015.

\bibitem[Blanchette et~al.(2011)Blanchette, Weber, Batty, Owens, and
  Sarkar]{blanchette+11}
J.~C. Blanchette, T.~Weber, M.~Batty, S.~Owens, and S.~Sarkar.
\newblock Nitpicking {C++} concurrency.
\newblock In \emph{PPDP}, 2011.

\bibitem[Dijkstra(2002)]{dijkstra02}
E.~W. Dijkstra.
\newblock Cooperating sequential processes (1965).
\newblock In P.~Brinch~Hansen, editor, \emph{The Origin of Concurrent
  Programming}, pages 65--138. Springer, 2002.

\bibitem[Flanagan and Godefroid(2005)]{flanagan+05a}
C.~Flanagan and P.~Godefroid.
\newblock Dynamic partial-order reduction for model checking software.
\newblock In \emph{POPL}, 2005.

\bibitem[Flur et~al.(2016)Flur, Gray, Pulte, Sarkar, Sezgin, Maranget, Deacon,
  and Sewell]{flur+16}
S.~Flur, K.~E. Gray, C.~Pulte, S.~Sarkar, A.~Sezgin, L.~Maranget, W.~Deacon,
  and P.~Sewell.
\newblock Modelling the {ARMv8} architecture, operationally: {C}oncurrency and
  {ISA}.
\newblock In \emph{POPL}, 2016.

\bibitem[Gaster et~al.(2015)Gaster, Hower, and Howes]{gaster+15}
B.~R. Gaster, D.~R. Hower, and L.~Howes.
\newblock {HRF-Relaxed}: Adapting {HRF} to the complexities of industrial
  heterogeneous memory models.
\newblock \emph{ACM Transactions on Architecture and Code Optimization}, 2015.

\bibitem[Hower et~al.(2014)Hower, Beckmann, Gaster, Hechtman, Hill, Reinhardt,
  and Wood]{hower+14}
D.~R. Hower, B.~M. Beckmann, B.~R. Gaster, B.~A. Hechtman, M.~D. Hill, S.~K.
  Reinhardt, and D.~A. Wood.
\newblock Adapting data-race-free memory consistency for heterogeneous systems.
\newblock In \emph{ASPLOS}, 2014.

\bibitem[{Intel Developer Zone}(2014)]{intel-developer-zone14}
{Intel Developer Zone}.
\newblock \emph{{OpenCL} 2.0 is here!}, 2014.
\newblock URL \url{https://software.intel.com/en-us/forums/opencl/topic/
  531074}.

\bibitem[ISO/IEC(2011{\natexlab{a}})]{c++11}
ISO/IEC.
\newblock \emph{Programming languages -- {C++}}.
\newblock International standard 14882:2011, 2011{\natexlab{a}}.

\bibitem[ISO/IEC(2011{\natexlab{b}})]{c11}
ISO/IEC.
\newblock \emph{Programming languages -- {C}}.
\newblock International standard 9899:2011, 2011{\natexlab{b}}.

\bibitem[ISO/IEC(2014)]{c++14}
ISO/IEC.
\newblock \emph{Programming languages -- {C++}}.
\newblock International standard 14882:2014, 2014.

\bibitem[{Khronos Group}(2015)]{opencl21}
{Khronos Group}.
\newblock \emph{The {OpenCL} Specification}.
\newblock Version 2.1, Revision 8, 2015.

\bibitem[{Khronos Group News Archives}(2014)]{khronos-group-news-archives14}
{Khronos Group News Archives}.
\newblock \emph{Freescale to spark innovation and open development for
  autonomous driving systems with {OpenCL}}, 2014.
\newblock URL \url{https://www.khronos.org/news/archives/2014/11}.

\bibitem[Lahav et~al.(2016)Lahav, Giannarakis, and Vafeiadis]{lahav+16}
O.~Lahav, N.~Giannarakis, and V.~Vafeiadis.
\newblock Taming release-acquire consistency.
\newblock In \emph{POPL}, 2016.

\bibitem[Lamport(1979)]{lamport79}
L.~Lamport.
\newblock How to make a multiprocessor computer that correctly executes
  multiprocess programs.
\newblock \emph{IEEE Transactions on Computers}, C-28\penalty0 (9), 1979.

\bibitem[Morriset et~al.(2013)Morriset, Pawan, and Zappa~Nardelli]{morriset+13}
R.~Morriset, P.~Pawan, and F.~Zappa~Nardelli.
\newblock Compiler testing via a theory of sound optimisations in the
  {C11/C++11} memory model.
\newblock In \emph{PLDI}, 2013.

\bibitem[Mulligan et~al.(2014)Mulligan, Owens, Gray, Ridge, and
  Sewell]{mulligan+14}
D.~P. Mulligan, S.~Owens, K.~E. Gray, T.~Ridge, and P.~Sewell.
\newblock Lem: reusable engineering of real-world semantics.
\newblock In \emph{ICFP}, 2014.

\bibitem[Norris and Demsky(2013)]{norris+13}
B.~Norris and B.~Demsky.
\newblock {CDSChecker}: Checking concurrent data structures written with
  {C/C++} atomics.
\newblock In \emph{OOPSLA}, 2013.

\bibitem[Orr et~al.(2015)Orr, Che, Yilmazer, Beckmann, Hill, and Wood]{orr+15}
M.~S. Orr, S.~Che, A.~Yilmazer, B.~M. Beckmann, M.~D. Hill, and D.~A. Wood.
\newblock Synchronization using remote-scope promotion.
\newblock In \emph{ASPLOS}, 2015.

\bibitem[Owens et~al.(2009)Owens, Sarkar, and Sewell]{owens+09}
S.~Owens, S.~Sarkar, and P.~Sewell.
\newblock A better x86 memory model: x86-{TSO}.
\newblock In \emph{TPHOLs}, 2009.

\bibitem[Sarkar et~al.(2011)Sarkar, Sewell, Alglave, Maranget, and
  Williams]{sarkar+11}
S.~Sarkar, P.~Sewell, J.~Alglave, L.~Maranget, and D.~Williams.
\newblock Understanding {POWER} multiprocessors.
\newblock In \emph{PLDI}, 2011.

\bibitem[Shasha and Snir(1988)]{shasha+88}
D.~Shasha and M.~Snir.
\newblock Efficient and correct execution of parallel programs that share
  memory.
\newblock \emph{TOPLAS}, 10\penalty0 (2), 1988.

\bibitem[Steuwer and Gorlatch(2013)]{steuwer+13}
M.~Steuwer and S.~Gorlatch.
\newblock High-level programming for medical imaging on multi-{GPU} systems
  using the {SkelCL} library.
\newblock In \emph{ICCS}, 2013.

\bibitem[Stone and Fitzgerald(1995)]{stone+95}
J.~M. Stone and R.~P. Fitzgerald.
\newblock Storage in the {PowerPC}.
\newblock In \emph{IEEE Micro}, 1995.

\bibitem[Tarski(1941)]{tarski41}
A.~Tarski.
\newblock On the calculus of relations.
\newblock \emph{Journal of Symbolic Logic}, 6\penalty0 (3):\penalty0 73--89,
  1941.

\bibitem[Turon et~al.(2014)Turon, Vafeiadis, and Dreyer]{turon+14}
A.~Turon, V.~Vafeiadis, and D.~Dreyer.
\newblock {GPS}: Navigating weak memory with ghosts, protocols, and separation.
\newblock In \emph{OOPSLA}, 2014.

\bibitem[Vafeiadis and Narayan(2013)]{vafeiadis+13}
V.~Vafeiadis and C.~Narayan.
\newblock Relaxed separation logic: A program logic for {C11} concurrency.
\newblock In \emph{OOPSLA}, 2013.

\bibitem[Vafeiadis et~al.(2015)Vafeiadis, Balabonski, Chakraborty, Morisset,
  and Zappa~Nardelli]{vafeiadis+15}
V.~Vafeiadis, T.~Balabonski, S.~Chakraborty, R.~Morisset, and
  F.~Zappa~Nardelli.
\newblock Common compiler optimisations are invalid in the {C11} memory model
  and what we can do about it.
\newblock In \emph{POPL}, 2015.

\bibitem[\v{S}ev\v{c}\'{i}k and Aspinall(2008)]{sevcik+08}
J.~\v{S}ev\v{c}\'{i}k and D.~Aspinall.
\newblock On validity of program transformations in the {Java} memory model.
\newblock In \emph{ECOOP}, 2008.

\bibitem[Wickerson et~al.(2015)Wickerson, Batty, Donaldson, and
  Beckmann]{wickerson+15a}
J.~Wickerson, M.~Batty, A.~F. Donaldson, and B.~M. Beckmann.
\newblock Remote-scope promotion: clarified, rectified, and verified.
\newblock In \emph{OOPSLA}, 2015.

\bibitem[Williams(2012)]{williams12}
A.~Williams.
\newblock \emph{C++ Concurrency in Action}.
\newblock Manning, 2012.

\end{thebibliography}

\end{document}

\clearpage

\section{Rules for SC atomics in C11}

\subsection{Original}
The following text is reproduced verbatim from the C11 standard~\citec{7.17.3, paragraphs 6 and 9--11}{}.

\begin{shaded*}
\begin{wordcounting}
\begin{enumerate}
\setcounter{enumi}{5}
\item There shall be a single total order $S$ on all
"memory\_order\_seq\_cst" operations, consistent with the ``happens
before'' order and modification orders for all affected locations,
such that each "memory\_order\_seq\_cst" operation $B$ that loads a
value from an atomic object $M$ observes one of the following values:
\begin{itemize}
\item the result of the last modification $A$ of $M$ that precedes $B$
in $S$, if it exists, or
\item if $A$ exists, the result of some modification of $M$ in the
visible sequence of side effects with respect to $B$ that is not
"memory\_order\_seq\_cst" and that does not happen before $A$, or
\item if $A$ does not exist, the result of some modification of $M$ in
the visible sequence of side effects with respect to $B$ that is not
"memory\_order\_seq\_cst".
\end{itemize}
\end{enumerate}
\noindent [\dots]
\begin{enumerate}
\setcounter{enumi}{8}
\item For an atomic operation $B$ that reads the value of an atomic
object $M$, if there is a "memory\_order\_seq\_cst" fence $X$
sequenced before $B$, then $B$ observes either the last
"memory\_order\_seq\_cst" modification of $M$ preceding $X$ in the
total order $S$ or a later modification of $M$ in its modification
order.
\item For atomic operations $A$ and $B$ on an atomic object $M$, where
$A$ modifies $M$ and $B$ takes its value, if there is a
"memory\_order\_seq\_cst" fence $X$ such that $A$ is sequenced before
$X$ and $B$ follows $X$ in $S$, then $B$ observes either the effects
of $A$ or a later modification of $M$ in its modification order.
\item For atomic operations $A$ and $B$ on an atomic object $M$, where
$A$ modifies $M$ and $B$ takes its value, if there are
"memory\_order\_seq\_cst" fences $X$ and $Y$ such that $A$ is
sequenced before $X$, $Y$ is sequenced before $B$, and $X$ precedes
$Y$ in $S$, then $B$ observes either the effects of $A$ or a later
modification of $M$ in its modification order.
\end{enumerate}
\end{wordcounting}
\reportwordcount{\ColdWORDS}{41.2}
\end{shaded*}

For reference, we include with this passage its Flesch--Kincaid (FK) reading ease score.\footnote{\url{http://www.readability-score.com}} A
higher score indicates easier readability. Scores usually range
between 0 and 100.

\newpage
\subsection{Our proposal}
\label{appx:c11_proposal}

Section~\ref{sec:sc-simp} presented our proposal for simplifying the
sequential consistency axioms in the C11 model. We give here our
suggestion for how the specification document can be rephrased to
accommodate our proposal, while maintaining the prose style used
throughout the rest of the document.

Specifically, the paragraphs quoted above can be removed
and replaced with the following three:

\begin{shaded*}
\begin{wordcounting}
\begin{enumerate}
\item A value computation $A$ of an object $M$ \emph{reads before} a
side effect $B$ on $M$ if $A$ and $B$ are different operations and $B$
follows, in the modification order of $M$, the side effect that $A$
observes.

\item If $X$ reads before $Y$, or happens before $Y$, or precedes $Y$
in modification order, then $X$ (as well as any fences sequenced
before $X$) is \emph{SC-before} $Y$ (as well as any fences sequenced
after $Y$).
 
\item If $A$ is SC-before $B$, and $A$ and $B$ are both {\tt
memory\_}{\tt order\_seq\_cst}, then $A$ is
\emph{restricted-SC-before} $B$.

\item There must be no cycles in restricted-SC-before.
\end{enumerate}
\end{wordcounting}
\reportwordcount{\CnewWORDS}{73.1}
\end{shaded*}

\clearpage
\section{Rules for SC atomics in OpenCL}

\subsection{Original}

The following text is reproduced verbatim from the OpenCL 2.1
standard~\citecl{51}{14}{52}{13}. We exclude the citations to the C11 standard. 

\begin{shaded*}
\begin{wordcounting}
\noindent If one of the following two conditions holds:
\begin{itemize}
\item All "memory\_order\_seq\_cst" operations have the scope
"memory\_scope\_all\_svm\_devices" and all affected memory locations
are contained in system allocations or fine grain SVM buffers with
atomics support
\item All "memory\_order\_seq\_cst" operations have the scope
"memory\_scope\_device" and all affected memory locations are not
located in system allocated regions or fine-grain SVM buffers with
atomics support
\end{itemize}
then there shall exist a single total order $S$ for all
"memory\_order\_seq\_cst" operations that is consistent with the
modification orders for all affected locations, as well as the
appropriate global-happens-before and local-happens-before orders for
those locations, such that each "memory\_order\_seq\_cst" operation
$B$ that loads a value from an atomic object $M$ in global or local
memory observes one of the following values:
\begin{itemize}
\item the result of the last modification $A$ of $M$ that precedes $B$
in $S$, if it exists, or
\item if $A$ exists, the result of some modification of $M$ in the
visible sequence of side effects with respect to $B$ that is not
"memory\_order\_seq\_cst" and that does not happen before $A$, or
\item if $A$ does not exist, the result of some modification of $M$ in
the visible sequence of side effects with respect to $B$ that is not
"memory\_order\_seq\_cst".
\end{itemize}
{[\dots]}\par
\noindent If the total order $S$ exists, the following rules hold:
\begin{itemize}
\item For an atomic operation $B$ that reads the value of an atomic
object $M$, if there is a "memory\_order\_seq\_cst" fence $X$
sequenced-before $B$, then $B$ observes either the last
"memory\_order\_seq\_cst" modification of $M$ preceding $X$ in the
total order $S$ or a later modification of $M$ in its modification
order.
\item For atomic operations $A$ and $B$ on an atomic object $M$, where
$A$ modifies $M$ and $B$ takes its value, if there is a
"memory\_order\_seq\_cst" fence $X$ such that $A$ is sequenced-before
$X$ and $B$ follows $X$ in $S$, then $B$ observes either the effects
of $A$ or a later modification of $M$ in its modification order.
\item For atomic operations $A$ and $B$ on an atomic object $M$, where
$A$ modifies $M$ and $B$ takes its value, if there are
"memory\_order\_seq\_cst" fences $X$ and $Y$ such that $A$ is
sequenced- before $X$, $Y$ is sequenced-before $B$, and $X$ precedes
$Y$ in $S$, then $B$ observes either the effects of $A$ or a later
modification of $M$ in its modification order. 
\item For atomic operations $A$ and $B$ on an atomic object $M$, if
there are "memory\_order\_seq\_cst" fences $X$ and $Y$ such that $A$
is sequenced-before $X$, $Y$ is sequenced-before $B$, and $X$ precedes
$Y$ in $S$, then $B$ occurs later than $A$ in the modification order
of $M$.
\end{itemize}
\end{wordcounting}
\reportwordcount{\OoldWORDS}{-22.0}
\end{shaded*}

\newpage
\subsection{Our proposal}
\label{appx:opencl_proposal}

Section~\ref{sec:opencl_sc} presented our proposal for simplifying the
sequential consistency axioms in the OpenCL model. We give here our
suggestion for how the specification document can be rephrased to
accommodate our proposal, while maintaining the prose style used
throughout the rest of the document.

Specifically, the paragraphs quoted above can be removed and replaced
with the following three. 

\begin{shaded*}
\begin{wordcounting}
\begin{enumerate}
\item A value computation $A$ of an object $M$ \emph{reads before} a
side effect $B$ on $M$ if $A$ and $B$ are different operations and $B$
follows, in the modification order of $M$, the side effect that $A$
observes.

\item If $X$ reads before $Y$, or global happens before $Y$, or local
happens before $Y$, or precedes $Y$ in modification order, then $X$
(as well as any fences sequenced before $X$) is \emph{SC-before} $Y$
(as well as any fences sequenced after $Y$).
 
\item If $A$ is SC-before $B$, and $A$ and $B$ are both {\tt
memory\_}{\tt order\_seq\_cst}, and $A$ and $B$ have inclusive scopes,
then $A$ is \emph{restricted-SC-before} $B$.

\item There must be no cycles in restricted-SC-before.
\end{enumerate}
\end{wordcounting}
\reportwordcount{\OnewWORDS}{71.0}
\end{shaded*}

The only departures from our proposal for the C11 memory model
(Sec.~\ref{appx:c11_proposal}) are the requirement of inclusive
scopes, and the splitting of happens-before into its global and local
versions.

\end{document}